\documentclass[a4paper]{scrartcl}
\usepackage{amsmath}
\usepackage{amsthm}
\usepackage{amssymb}
\usepackage{tikz}

\usepackage{enumitem}
\usepackage{ifthen}

\usepackage{xurl}
\usepackage[natbib=true, isbn=false, date=year, citetracker, maxcitenames=2, bibstyle=synphonie]{biblatex}
\AtEveryCitekey{%
  \ifciteseen{}{\defcounter{maxnames}{6}}%
}

\usepackage{authblk}

\usepackage[colorlinks,linkcolor={red!50!black},citecolor={blue!50!black},urlcolor={blue!80!black}]{hyperref}
\usepackage[capitalize,noabbrev]{cleveref}

\title{An Equilibrium Model for Schedule-Based Transit Networks with Hard Vehicle Capacities}

\author[1]{Tobias Harks}
\author[2]{Sven Jäger}
\author[1]{Michael Markl}
\author[3]{Philine Schiewe}
\affil[1]{University of Passau, Germany}
\affil[2]{RPTU Kaiserslautern-Landau, Germany}
\affil[3]{Aalto University, Finland}

\date{}

\newtheorem{theorem}{Theorem}
\newtheorem{remark}[theorem]{Remark}
\newtheorem{lemma}[theorem]{Lemma}
\newtheorem{proposition}[theorem]{Proposition}
\newtheorem{corollary}[theorem]{Corollary}
\newtheorem{claim}[theorem]{Claim}
\theoremstyle{definition}
\newtheorem{definition}[theorem]{Definition}
\newtheorem{notation}[theorem]{Notation}
\newtheorem{example}[theorem]{Example}

\newtheoremstyle{problemstyle}%
  {}%
  {}%
  {}%
  {}%
  {\bfseries}%
  {.}%
  {\newline}%
  {}%

\theoremstyle{problemstyle}
\newtheorem{innerproblem}{Problem}

\newenvironment{problem}[1]%
  {\renewcommand{\theinnerproblem}{(#1)}%
   \innerproblem}%
  {\endinnerproblem}

\crefname{innerproblem}{Problem}{Problems}

\newenvironment{proofClaim}{\proof}{\endproof}

\newcommand{\BFzero}{\mathbf{0}}

\DeclareMathOperator*{\argmin}{arg\,min}

\crefname{algocf}{Algorithm}{Algorithms}
\Crefname{algocf}{Algorithm}{Algorithms}

\newcommand{\optDisplay}[1]{\[ #1 \]}

\usepackage{mathtools}
\usetikzlibrary{calc}

\usepackage{pgfplots}
\pgfplotsset{compat=1.18}
\usetikzlibrary{pgfplots.groupplots}

\tikzset{
 >=latex
}

\usepackage{placeins}
\usepackage[group-separator={,},group-minimum-digits=4]{siunitx}
\usepackage{xspace}
\usepackage{nicefrac}
\usepackage{svg}
\usepackage[ruled, noend, linesnumbered]{algorithm2e}
\usepackage{booktabs}
\usepackage{tabularx}
\usepackage[outline]{contour}
\usepackage{thmtools}

\newcommand{\oxc}{,\xspace}
\newcommand*{\R}{\mathbb R}
\newcommand*{\Z}{\mathbb Z}
\newcommand*{\N}{\mathbb N}
\newcommand*{\bigO}{\mathcal O}

\DeclarePairedDelimiter{\abs}{\lvert}{\rvert}
\newcommand*{\norm}[1]{\left\lVert #1 \right\rVert}
\newcommand*{\smallnorm}[1]{\lVert #1 \rVert}
\DeclarePairedDelimiterX{\scalar}[2]{\langle}{\rangle}{#1, #2}

\newcommand*{\demand}{Q} %
\newcommand*{\stations}{S}
\newcommand*{\vehicles}{Z}
\newcommand*{\nodetime}{\theta}
\newcommand*{\capacity}{\nu}
\newcommand*{\pathsWithOutside}{\mathcal{P}}
\newcommand*{\pathsWithoutOutside}{\mathcal{P}^\circ}
\newcommand*{\outEdges}[1]{\delta^-_v}
\newcommand*{\inEdges}[1]{\delta^+_v}
\newcommand*{\addmEpsDev}{admissible \epsDev} 
\newcommand*{\admissibleDeviations}{D}
\newcommand*{\epsDev}{$\varepsilon$-deviation}

\newcommand{\VI}{\textrm{VI}}
\newcommand{\QVI}{\textrm{QVI}} %
\newcommand{\outopt}{p^{\mathrm{out}}} %
\newcommand{\avPaths}{A} %

\newcommand*{\demFeasFlows}{\mathcal F_\demand} %

\newcommand*{\departureTimes}{\Theta}
\newcommand*{\targetTime}{T}

\newcommand*{\patharrival}{\textnormal{arr}}

\newcommand{\ordinalth}{th}

\newcommand{\sbTime}[2]{\ensuremath{#1^{#2}}}

\AtBeginDocument{ 
\definecolor{customgray}{RGB}{221,221,221}
\definecolor{color1}{HTML}{190085}
\definecolor{color2}{HTML}{f21d1d}
\definecolor{color3}{HTML}{21c021}
\definecolor{color4}{HTML}{ee96ff}
\newcommand{\colourone}{blue\xspace}
\newcommand{\colourtwo}{red\xspace}
\newcommand{\colourthree}{green\xspace}
\newcommand{\colourfour}{pink\xspace}
}

\newcommand{\lifelineCompact}[2]{
    \draw[dashed] (#1, \height - 0.5) -- (#1, -0.5);
    \node[anchor=base] at (#1, \height + 0.1) {#2};
}

\newcommand{\lifelineCompactC}[2]{
    \draw[dashed, |-|] (#1, \height - 0.5) -- (#1, -0.5);
    \node[anchor=base] at (#1, \height + 0.1) {#2};
}
\newcommand{\waitnode}[3]{
    \node[fill,circle,inner sep=0.05cm] (#1) at (#2, \height - #3) {};
}
\newcommand{\departurenode}[4]{
    \node[fill,circle,inner sep=0.05cm] (#1) at (#2 + #4*\depxoffset, \height - #3) {};
}
\newcommand{\arrivalnode}[4]{
    \node[fill,circle,inner sep=0.05cm] (#1) at (#2 + #4*\depxoffset, \height - #3) {};
}

\newcommand{\lifelineVerbose}[2]{
    \fill[fill=customgray] (#1 - 0.1, -0.5) rectangle (#1 + 0.1, \height - 0.5);
    \node[anchor=base] at (#1, \height + 0.1) {#2};
}

\DeclareDocumentCommand{\qq}{m}{“#1”}

\usepackage{todonotes}
\makeatletter
\renewcommand{\todo}[2][]{\tikzexternaldisable\@todo[#1]{#2}\tikzexternalenable}
\makeatother

\begin{document}

    \maketitle

    \begin{abstract}
        Modelling passenger assignments in public transport networks is a fundamental task for city planners, especially when deliberating
        network infrastructure decisions.
        A key aspect of a realistic model is to integrate passengers' selfish routing behaviour under limited vehicle capacities.
        We formulate a side-constrained user equilibrium model in a schedule-based transit network, where passengers are modelled via a continuum of non-atomic agents that travel from their origin to their destination.
        An agent's route may comprise several rides along given lines, each using vehicles with hard loading capacities.
        We give a characterization of (side-constrained) user equilibria via a quasi-variational inequality and prove their existence for fixed departure times by generalizing a well-known result of Bernstein and Smith~(Transp.\,Sci., 1994).
        We further derive a polynomial time algorithm for single-commodity instances with fixed departure times.
        For the multi-commodity case with departure time choice, we show that deciding whether an equilibrium exists is NP-hard, and we devise an exponential-time algorithm that computes an equilibrium if it exists, and signals non-existence otherwise.
        Using our quasi-variational characterization, we formulate a heuristic for computing multi-commodity user equilibria in practice, which is tested on multiple real-world instances.
        In terms of social cost, the computed user-equilibria are quite efficient compared to a system optimum.
    \end{abstract}

    \clearpage
    \tableofcontents
    \clearpage

\section{Introduction}
In the domain of public transport,
models describing the assignment of passengers over a transit network are crucial for infrastructure planners to
understand congestion phenomena and assess possible investments into the infrastructure.
With new advances in technology, the information available to passengers on effective schedules (adjusted by real-time delays), capacities\oxc and utilization of vehicles is ever-increasing.
As a result, the routing behaviour of the passengers is affected by their (close to) full information on the current and future network state.

The existing approaches for modelling transit networks can roughly
be categorized into \emph{frequency-based}
and \emph{schedule-based} models, see \cite{GentileFHCN16,Fu2012} for a survey.
The former model class operates with line frequencies
and implicitly defines resulting travel times
and capacities of lines and vehicles, cf.~\citep{Spiess1989,CeaF93,WuFM94,Bouzaiene-AyariGN98,CominettiC01,CepedaCF06,Larrain21}.
With variations in the demand profile during peak hours, the frequency-based approach only leads to approximate vehicle loads,
with the error increasing as variability grows.
In contrast, schedule-based approaches are more fine-grained and capable of explicitly modelling irregular timetablesof lines.
They are usually based on a \emph{time-expanded transit network} derived from the physical transit network and augmented by (artificial) edges
such as waiting, boarding, alighting, dwelling\oxc and
driving edges to connect different stations.
\Cref{fig:intro} illustrates this construct, also known as diachronic graph~\citep{NuzzoloR96} or space-time network~\citep{CarraresiMP96}.

    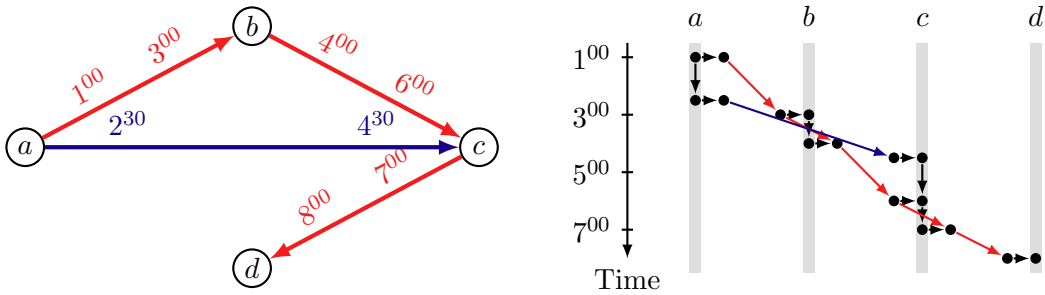
\begin{figure}[b]
        \centering
        \begin{minipage}[c]{0.5\textwidth}%
    \centering
    \begin{tikzpicture}[x=3cm, y=1.6cm, thick]
        \foreach[count=\i] \v in {a,b,c,d}
            \node[draw, circle, inner sep=0pt, minimum size=3ex] (\v) at ({270-90*\i}:1) {$\v$};
        \foreach \i/\j/\d/\a in {a/b/\sbTime{1}{00}/\sbTime{3}{00}, b/c/\sbTime{4}{00}/\sbTime{6}{00}, c/d/\sbTime{7}{00}/\sbTime{8}{00}}
            \draw[ultra thick, ->, color2] (\i) -- node[above, sloped, pos=0.3] {\d} node[above, sloped, pos=0.7] {\a} (\j);
        \draw[ultra thick, ->, color1] (a) -- node[above, pos=0.2] {\sbTime{2}{30}} node[above, pos=0.8] {\sbTime{4}{30}} (c);
    \end{tikzpicture}
\end{minipage}%
\begin{minipage}[c]{0.5\textwidth}%
    \centering
    \begin{tikzpicture}[y=0.38cm, x=0.75cm, thick]
        \newcommand{\height}{8}
        \newcommand{\depxoffset}{0.5}

        \def\yAxisCoord{-1.2}
        \draw[->] (\yAxisCoord,\height - 0.5) -- (\yAxisCoord, 0) node[below] {Time};
        \foreach \y in {1,3,5,7}
            \draw (\yAxisCoord + 0.1,\height - \y) -- (\yAxisCoord -0.1,\height - \y) node[left] {\sbTime{\y}{00}};

        \lifelineVerbose{0}{$a$}
        \lifelineVerbose{2}{$b$}
        \lifelineVerbose{4}{$c$}
        \lifelineVerbose{6}{$d$}
                \waitnode{W_0_1}{0}{1}
        \waitnode{W_4_4d5}{4}{4.5}
        \waitnode{W_2_4}{2}{4}
        \waitnode{W_6_8}{6}{8}
        \waitnode{W_4_6}{4}{6}
        \waitnode{W_2_3}{2}{3}
        \waitnode{W_0_2d5}{0}{2.5}
        \waitnode{W_4_7}{4}{7}
        \draw[->] (W_0_1) -- (W_0_2d5);
        \draw[->] (W_4_4d5) -- (W_4_6);
        \draw[->] (W_4_6) -- (W_4_7);
        \draw[->] (W_2_3) -- (W_2_4);
        \departurenode{D_0_1}{0}{1}{1}
        \arrivalnode{A_2_3}{2}{3}{-1}
        \draw[->] (W_0_1) -- (D_0_1);
        \draw[->, draw=color2] (D_0_1) -- (A_2_3);
        \draw[->] (A_2_3) -- (W_2_3);
        \departurenode{D_2_4}{2}{4}{1}
        \draw[->, draw=color2] (A_2_3) -- (D_2_4);
        \arrivalnode{A_4_6}{4}{6}{-1}
        \draw[->] (W_2_4) -- (D_2_4);
        \draw[->, draw=color2] (D_2_4) -- (A_4_6);
        \draw[->] (A_4_6) -- (W_4_6);
        \departurenode{D_4_7}{4}{7}{1}
        \draw[->, draw=color2] (A_4_6) -- (D_4_7);
        \arrivalnode{A_6_8}{6}{8}{-1}
        \draw[->] (W_4_7) -- (D_4_7);
        \draw[->, draw=color2] (D_4_7) -- (A_6_8);
        \draw[->] (A_6_8) -- (W_6_8);
        \departurenode{D_0_2d5}{0}{2.5}{1}
        \arrivalnode{A_4_4d5}{4}{4.5}{-1}
        \draw[->] (W_0_2d5) -- (D_0_2d5);
        \draw[->, draw=color1] (D_0_2d5) -- (A_4_4d5);
        \draw[->] (A_4_4d5) -- (W_4_4d5);
    \end{tikzpicture}
\end{minipage}%
        \caption{Two scheduled vehicle trips in the physical network (left) and their representation in the time-expanded transit network (right).}\label{fig:intro}
    \end{figure}

An assignment of passengers to paths in this network encompasses their entire travel strategies, including line changes, waiting times, etc.
It corresponds to a path-based multi-commodity network flow satisfying all demand and supply.
A key challenge in the analysis of such a schedule-based model is the integration of \emph{strategic behaviour} of passengers,
opting for shortest routes, and the \emph{limited vehicle capacity},
which restricts the number of passengers able to use a vehicle at any point in time.
If a vehicle is already at capacity, further passengers might not be able to enter this vehicle at the next boarding station, which can make their (shortest) route infeasible.
On the other hand, the passengers already in the vehicle are not affected by the passengers wishing to board.

A key issue of such a capacitated model is to choose the right equilibrium concept.
Consider for instance the simple example in \Cref{fig:intro}, and suppose that the vehicles operating the \colourone and the \colourtwo line have a capacity of $1$ unit each.
A demand volume of $2$ units start their trip at node~$a$ at time \sbTime{1}{00}, and all particles want to travel as fast as possible to the destination node $c$.
The \colourone line arrives at $4^{30}$, while the \colourtwo line arrives at \sbTime{6}{00}.
Then there exists no capacity-feasible \emph{Wardrop equilibrium}~\citep{Wardrop52,CorreaS2011}, i.e., a flow only using quickest paths.

Most works in the literature deal with this non-existence by either assuming soft vehicle capacities (cf.~\citealt{Crisalli99,NuzzoloRC2001,NguyenPM01}) or by considering more general travel strategies and a probabilistic loading mechanism (cf.~\citealt{MarcotteN98,Kurauchi03,Marcotte04,HamdouchMN04}).
An alternative approach that inherently supports capacities are so-called \emph{side-constrained user equilibria}, for which different defintions have been proposed \citep{Daganzo77,Hearn1980,CorSS04}. In our study, we consider a definition based on \emph{admissible deviations} \citep{Smith84} for schedule-based time-expanded transit networks. Whether a deviation is admissible depends only on the available capacity of the vehicle when the passenger boards it, but not on whether capacity is exceeded on a later edge of the vehicle trip.
Hence, a path can be an available alternative for some user even if arbitrarily small deviations to that path make the resulting flow infeasible (for some other users).

As in \cite{NguyenPM01}, the priority of passengers in the vehicle can be modelled by expressing the capacity limitations using discontinuous costs on the boarding edges in the time-expanded transit network.
The resulting cost map is not separable, and it turns out that it does not satisfy the regularity conditions imposed by \citet{BernsteinS94} to prove existence of equilibria.

To model realistic passenger behaviour, we extend our assumptions about user preferences in two ways: First, users are generally concerned not only with travel time but also with departing or arriving at a preferred time, and they will choose their departure time accordingly. This is taken into account by using general personal costs that incorporate penalty terms for the deviation of the arrival time from the desired time. Second, users may have a personal limit on their travelling cost. If the travel cost exceeds this limit, they may cancel the trip or opt for an alternative mode of transport, such as a private car. This is known as so-called \emph{elastic demand}, i.e., the demand of the system depends on the offer. This elastic demand model is quite standard in the transportation science literature, see  \cite{Wu03} and references therein.

\subsection{Our Contribution}

We define a user equilibrium for schedule-based time-expanded transit networks using the notion of admissible deviations.
For a given flow, an \emph{admissible $\varepsilon$-deviation} corresponds to shifting an $\varepsilon$-amount of flow from a path $p$ to another path $q$ without exceeding the capacity of any boarding edge along~$q$.
A feasible flow is a \emph{side-constrained user equilibrium} if there are no improving admissible $\varepsilon$-deviations for arbitrarily small $\varepsilon$.
We summarize our contribution as follows.
\begin{enumerate}
    \item
        We characterize side-constrained user equilibria in schedule-based time-expanded tansit networks by a quasi-variational inequality defined over the set of admissible deviations (\Cref{thm:QVI}) and as equilibria in the sense of \citet{BernsteinS94} (subsequently referred to as \emph{BS-equilibria}) in an extended network with discontinuous user cost functions (\cref{thm:BS-equilibrium iff SC-equilibrium}).
    \item
        We study the central question of the existence of side-constrained equilibria. While existence is not guaranteed if departure times can be chosen freely (\cref{example:non-existence-dtc}), we can guarantee existence for the important special case of fixed departure times.
        For this, we generalize a result of \citet{BernsteinS94} who showed that BS-equilibria exist for \emph{regular} cost maps.
        While our cost map does not fall into the category of \emph{regular} cost maps, we introduce a more general condition for cost maps, which we term \emph{weakly regular}.
        We prove that BS-equilibria do exist for weakly regular cost maps (\cref{thm:main}) and that the cost maps in schedule-based time-expanded transit networks are weakly regular for fixed departure times (\cref{thm:existence-metro}). The general existence result for weakly regular cost maps might be of interest also for other traffic models.
    \item 
        We then turn to the computation of user equilibria.
        For single-commodity time-expanded networks with fixed departure times, we present an algorithm that computes a BS-equilibrium in quadratic time relative to the number of edges of the input graph (\Cref{thm:single_comp}).
        For multi-commodity networks, we show that it is NP-hard to decide whether user equilibria (with departure time choice) exist; NP-hardness also applies to related decision problems even when restricting to instances with fixed departure times.
        Lastly, we give an exact finite-time algorithm for the multi-commodity scenario.
        As this algorithm is too slow for practical computations, we further develop a heuristic based on our quasi-variational inequality formulation.
        It starts with an arbitrary feasible flow and updates this flow along elementary admissible deviations in the sense of \Cref{thm:QVI}.
    \item
        Finally, we test our heuristic on realistic instances drawn from the TimPassLib \citep{SchieweGL23} database.
        While even approximate equilibria are not guaranteed to exist in the case of departure time choice, we find that the heuristic computes flows that are close to user equilibria in practice:
        It computes flows with a 99\ordinalth{} percentile equilibrium-approximation factor of up to \num{1.24} in 5 of 7 instances for the case of fixed departure times, and a factor of up to \num{2.48} in 5 of 7 instances for the case of departure time choice.
        Compared to a system optimum, which neglects equilibrium constraints and minimizes total travel cost,
        the total travel cost of the flows computed by the heuristic is at most 8\% higher in 6 of 7 instances and 20\% higher in the remaining instance.
\end{enumerate}

\paragraph{Comparison to the Conference Version.}

This paper is an extension of the conference paper~\citep{Harks2024}.
In comparison, this paper not only contains the fully-worked out proofs of all statements and a more detailed discussion of the related literature but also introduces several new results:
Firstly, we extend the model to incorporate departure time choice of users and generalize the characterization of user equilibria.
Secondly, to prove the existence of user equilibria in schedule-based time-expanded transit networks for fixed departure times, we present several insights on the structure of these user equilibria.
Thirdly, we show that the price of stability is unbounded for the considered model by presenting a small and concrete problem instance.
Based on this problem instance, we show that deciding whether a user equilibrium exists is NP-hard.
This hardness also applies to related decision problems even when restricting to instances with fixed departure times.
Furthermore, we present the algorithm for computing user equilibria for single commodities with fixed departure times and its correctness proof in detail; we also give an example for a multi-commodity instance for which the main assumption of the single-commodity algorithm fails.
Lastly, we provide a thorough discussion of the heuristic for multi-commodity user equilibria, present techniques to improve its performance, and conduct an extended computational study.

\subsection{Related Work}\label{related-work}

A large body of research deals with schedule-based transit assignment, see for example the two proceedings volumes \citep{WilsonN04,NuzzoloW09}.
We limit ourselves in the following literature review to schedule-based models that incorporate congestion; for frequency-based and uncongested models, see \citep{Fu2012,GentileFHCN16,GentileNSTC16}.
Most work uses the schedule-based time-expanded transit network as a modelling basis, which can be traced back to \citet{NuzzoloR96} and \citet{CarraresiMP96}.

\Citet{CarraresiMP96} consider a model with hard capacity constraints. They are interested in finding a transit assignment where the cost of every passenger is only a factor of $1+\varepsilon$ worse than the optimal cost in an uncongested network. Such a routing is only possible when the delays due to congestion are not too large. %
In heavily congested networks, %
passengers are satisfied with a route that is best possible under the given congestion conditions.
This is approximated in several papers \citep{Crisalli99,NguyenPM01,NuzzoloRC2001} by incorporating the vehicle capacities as continuous penalties representing the discomfort experienced by using an overcrowded edge.

\Citet{MarcotteN98} deal with hard capacities by defining the \emph{strategy} of an agent as preference orderings of outgoing edges at each node and by assuming a random loading mechanism for congested edges, where the probability of being able to enter an edge is proportional to the capacity and decreases with the number of agents desiring to traverse it as well.
Every passenger wants to minimize the \emph{expected} travel cost resulting from their strategy.
\Citet{Marcotte04} investigate the model further, including some computational experiments.
\Citet{Zimmermann2021} integrate this concept with Markovian traffic equilibria introduced by \Citet{Baillon2006} where perceived travel costs are subject to stochastic variations.
This leads to a loading mechanism based on choice probabilities between pairs consisting of a node and an availability vector for its outgoing edges; the authors suggest a heuristic for computing equilibria in this setting.

Instead of a stochastic loading mechanism, \citet{NguyenPM01} introduce a model in which the incoming edges of every departure node are ordered, and the outgoing driving edge is filled with passengers in the order of the edge through which they arrived.
This allows to model \emph{FIFO queues} of passengers aiming to board a vehicle.
Their model is closest to the one considered in this paper.
However, they do not compute equilibria for this model, but switch to an approximate model with continuous penalty terms for the computation. A similar assumption is made by \citet{Akamatsu2023}, who incorporate this aspect in the Markovian setting by using non-separable, differentiable cost functions that tend to infinity when approaching the capacity. For this model, the authors analyze existence, uniqueness\oxc and the global stability of the day-to-day dynamics.

\Citet{HamdouchMN04} combine this priority-based approach with the random loading described before, by assuming that the passengers already in the vehicle can always stay there, while the others take part in the random loading mechanism, introduced by \citeauthor{MarcotteN98}.
In this way, they model passengers mingling at stations.
Again, an agent's cost function is defined as the expected travel time for their chosen strategy, for which an according variational inequality has a solution due to the Ky Fan Inequality.
This model has been extended in multiple ways:
In \citep{Hamdouch2008}, the expected strategy costs are generalized by allowing early departure and late/early arrival penalties as well as a crowding discomfort.
\Citet{Hamdouch2011} differentiate between seated and standing passengers, affecting the discomfort (and thus, the expected strategy cost), and \citet{RochauNB12} account for risk-aversion.
Finally, \citet{Hamdouch14} incorporate uncertainties in link travel times modeling variations due to weather effects, incidents, etc., and \citet{Kumar23} take into account that transfers may be missed due to delays.
\Citet{Nuzzolo2012} study the computation of strategy profiles through learning in an iterative heuristic.
Similarly, \Citet{Patzner2024} propose an agent-based assignment method that considers vehicle capacities explicitly and that incorporates a learning mechanism.

Another approach without time-expanded graphs is pursued by \citet{PapolaFGM09} who consider a network combining scheduled public-transport edges and continuous pedestrian edges.
They define dynamic flows that traverse the public-transport edges in discrete chunks.
To find approximate equilibria, they employ the method of successive averages (MSA).
\Citet{GrafH23} study side-constrained equilibria for dynamic flows.
They consider a dynamic variant of BS-equilibria but give no existence result for them.

As mentioned in the introduction, different concepts have been proposed to model equilibria under side constraints. An early concept of side-constrained equilibria~\citep{Daganzo77,Hearn1980} simply adds side constraints to the Beckman-McGuire-Winsten (BMW) formulation~\citep{BMW56}, whose solutions correspond to Wardrop equilibria in the absence of side constraints. %
\Citet{Larsson95,Larsson99} show that under some natural conditions, %
solutions of this convex program have the property that for any used path there is no alternative path with available residual capacity and lower cost in the original cost function.
This \emph{extended Wardrop principle}~\citep{Marcotte04} is used as the definition of side-constrained equilibria by \citet{CorSS04}, who show that it may lead to strictly more equilibria than the side-constrained BMW formulation. In particular, it allows for equilibria where some particles could actually change from their used path~$p$ to a shorter path~$q$ that shares some saturated edges with $p$. This can be avoided by amending the definition of admissible deviations to allow deviations to alternative paths whenever the resulting flow still complies with the side constraints. This idea, originally formulated by \citet{Smith84}, corresponds to the equilibrium concept of \citet{BernsteinS94}, when the constraints are modelled using discontinuous edge cost functions, and leads to a set of equilibria that lies between the two equilibrium sets described above.

\section{Side-Constrained Equilibria for Schedule-Based Transit Networks}\label{sec:model}

We  first describe a schedule-based time-expanded network (cf.~\citealt{NuzzoloR96,CarraresiMP96}) and then formally define the side-constrained user equilibrium concept.

\subsection{The Time-Expanded Transit Network}

Consider a set of geographical stations $\stations$ (e.g., metro stations or bus stops) and a set of vehicle trips~$\vehicles$ (e.g., trips of metro trains or buses), specified by their sequence of served stations and adhering to a fixed, reliable timetable.
This timetable specifies the arrival and departure times at all stations of the trip, where the arrival time at a station is always strictly later than the departure time at the previous station.
Each vehicle trip $z\in \vehicles$ also has an associated capacity $\capacity_z$ which represents the maximum number of users the corresponding vehicle may hold at any time.
Throughout this work, we use the term \emph{vehicle} synonymously with \emph{vehicle trip}.

To represent the passengers' routes through the network, we construct a time-expanded directed acyclic graph $G=(V,E)$ with a time~$\nodetime(v)\in\R$ assigned to each node~$v\in V$.

There are three types of nodes:
An \emph{on-platform node} represents that a user is located on a station's platform where the user may board a vehicle or wait on the platform; we generate an on-platform node for each station~$s\in\stations$ and each time at which at least one vehicle departs or arrives in $s$.
A \emph{departure node} represents that the user is on a vehicle which is about to depart from a station; thus, we create a departure node for each vehicle $z\in\vehicles$ and each time~$\theta$ at which $z$ departs from a station~$s$.
Similarly, an \emph{arrival node} represents that the user is on a vehicle which has just arrived at a station; an arrival node is created for each vehicle $z\in\vehicles$ and time~$\theta$ at which $z$ arrives at a station~$s$.
 
\begin{figure}[t]
            \centering
        \newlength{\lenseXdim}
\newlength{\lenseXOffset}
    \setlength{\lenseXdim}{0.7cm}
    \setlength{\lenseXOffset}{9cm}

\begin{tikzpicture}[y=0.38cm, x=\lenseXdim, thick]
    \newcommand{\height}{8}
    \newcommand{\depxoffset}{0.5}

    \definecolor{lensbackground}{RGB}{255, 255, 255}
    \definecolor{focusbackground}{RGB}{255, 255, 255}
    \fill[rounded corners=8pt, fill=focusbackground] (3,\height - 5) rectangle (5, \height - 7.5);

    \def\yAxisCoord{-1.2}
    \draw[->] (\yAxisCoord,\height - 0.5) -- (\yAxisCoord, 0) node[below] {Time};
    \foreach \y in {1,3,5,7}
        \draw (\yAxisCoord + 0.1,\height - \y) -- (\yAxisCoord -0.1,\height - \y) node[left] {\sbTime{\y}{00}};

    \lifelineVerbose{0}{$a$}
    \lifelineVerbose{2}{$b$}
    \lifelineVerbose{4}{$c$}
    \lifelineVerbose{6}{$d$}
        \waitnode{W_0_1}{0}{1}
    \waitnode{W_4_4d5}{4}{4.5}
    \waitnode{W_2_4}{2}{4}
    \waitnode{W_6_8}{6}{8}
    \waitnode{W_4_6}{4}{6}
    \waitnode{W_2_3}{2}{3}
    \waitnode{W_0_2d5}{0}{2.5}
    \waitnode{W_4_7}{4}{7}
    \draw[->] (W_0_1) -- (W_0_2d5);
    \draw[->] (W_4_4d5) -- (W_4_6);
    \draw[->] (W_4_6) -- (W_4_7);
    \draw[->] (W_2_3) -- (W_2_4);
    \departurenode{D_0_1}{0}{1}{1}
    \arrivalnode{A_2_3}{2}{3}{-1}
    \draw[->] (W_0_1) -- (D_0_1);
    \draw[->, draw=color2] (D_0_1) -- (A_2_3);
    \draw[->] (A_2_3) -- (W_2_3);
    \departurenode{D_2_4}{2}{4}{1}
    \draw[->, draw=color2] (A_2_3) -- (D_2_4);
    \arrivalnode{A_4_6}{4}{6}{-1}
    \draw[->] (W_2_4) -- (D_2_4);
    \draw[->, draw=color2] (D_2_4) -- (A_4_6);
    \draw[->] (A_4_6) -- (W_4_6);
    \departurenode{D_4_7}{4}{7}{1}
    \draw[->, draw=color2] (A_4_6) -- (D_4_7);
    \arrivalnode{A_6_8}{6}{8}{-1}
    \draw[->] (W_4_7) -- (D_4_7);
    \draw[->, draw=color2] (D_4_7) -- (A_6_8);
    \draw[->] (A_6_8) -- (W_6_8);
    \departurenode{D_0_2d5}{0}{2.5}{1}
    \arrivalnode{A_4_4d5}{4}{4.5}{-1}
    \draw[->] (W_0_2d5) -- (D_0_2d5);
    \draw[->, draw=color1] (D_0_2d5) -- (A_4_4d5);
    \draw[->] (A_4_4d5) -- (W_4_4d5);

    \draw[rounded corners=8pt] (3,\height - 5) rectangle (5, \height - 7.5);

    \node[inner sep=0, outer sep=0] (FOCUS_NORTH_WEST) at ($(3, \height - 5) + (4pt, -1.1pt)$) {};
    \node[inner sep=0, outer sep=0] (FOCUS_SOUTH_WEST) at ($(3, \height - 7.5) + (6pt, 0.175pt)$) {};

    \begin{scope}[xshift=\lenseXOffset, yshift=0.38cm]
        \footnotesize
        \contourlength{0.11em}
        \newcommand{\nodelabelrotation}{35}
        \newcommand{\nodelabeldistance}{0}

        \begin{scope}
        \clip[rounded corners=16pt] (-6,\height - 0) rectangle (6, \height - 10);
        \fill[rounded corners=16pt, fill=lensbackground] (-6,\height - 0) rectangle (6, \height - 10);

        \fill[fill=customgray] (0 - 0.5, -2) rectangle (0 + 0.5, \height);

        \node[inner sep=0, outer sep=0] (LENS_NORTH_WEST) at ($(-6, \height - 0) + (8pt, -2.1pt)$) {};
        \node[inner sep=0, outer sep=0] (LENS_SOUTH_WEST) at ($(-6, \height - 10) + (15pt, 0pt)$) {};

        \begin{scope}[yshift=-0.19cm]
            \node[fill,circle,inner sep=0.05cm, outer sep=0.05cm] (LENS_W_4_6) at (0, \height - 3) {};
            \node[fill,circle,inner sep=0.05cm, outer sep=0.05cm] (LENS_W_4_7) at (0, \height - 7) {};

            \node[fill,circle,inner sep=0.05cm, outer sep=0.05cm] (LENS_ARR) at (-4, \height - 3) {};
            \node[rotate=\nodelabelrotation, above right=\nodelabeldistance of LENS_ARR, anchor=west] {\contour{lensbackground}{arrival}};

            \node[fill,circle,inner sep=0.05cm, outer sep=0.05cm] (LENS_DEP) at (4, \height - 7) {};
            \node[rotate=\nodelabelrotation, above right=\nodelabeldistance of LENS_DEP, anchor=west] {\contour{lensbackground}{departure}};

            \draw[->, draw=color2] (LENS_ARR) to[out=315, in=150] node[pos=0.27] {\contour{lensbackground}{dwelling}} (LENS_DEP);
            \draw[->, draw=color2] ($(LENS_ARR) + (-2, 3)$) -- (LENS_ARR) node[midway] {\contour{lensbackground}{driving}};
            \draw[draw=color2] (LENS_DEP) -- ($(LENS_DEP) + (2, -2)$) node[midway] {\contour{lensbackground}{driving}};

            \draw[->] (LENS_ARR) -- (LENS_W_4_6) node[midway] {\contour{lensbackground}{alighting}};
            \draw[->] ($(LENS_W_4_6) + (0, 3.5)$) -- (LENS_W_4_6) node[pos=0.33] {\contour{lensbackground}{waiting}};
            \draw[->] (LENS_W_4_6) -- (LENS_W_4_7) node[pos=0.33] {\contour{lensbackground}{waiting}};
            \draw[->] (LENS_W_4_7) -- (LENS_DEP) node[midway] {\contour{lensbackground}{boarding}};
            \node[rotate=\nodelabelrotation, above right=\nodelabeldistance of LENS_W_4_7, anchor=west] {\contour{lensbackground}{on-platform}};
            \node[rotate=\nodelabelrotation, above right=\nodelabeldistance of LENS_W_4_6, anchor=west] {\contour{lensbackground}{on-platform}};
        \end{scope}
        \end{scope}
        \draw[rounded corners=16pt] (-6,\height - 0) rectangle (6, \height - 10);
    \end{scope}

    \draw [-] (FOCUS_NORTH_WEST) -- (LENS_NORTH_WEST);
    \draw [-] (FOCUS_SOUTH_WEST) -- (LENS_SOUTH_WEST);

\end{tikzpicture}
        \caption{Visualization of the time-expanded transit network from \Cref{fig:intro}. Each station is represented by a vertical timeline. The driving and dwelling edges are coloured according to their vehicle trip.}\label{fig:intro-lens}
    \end{figure}
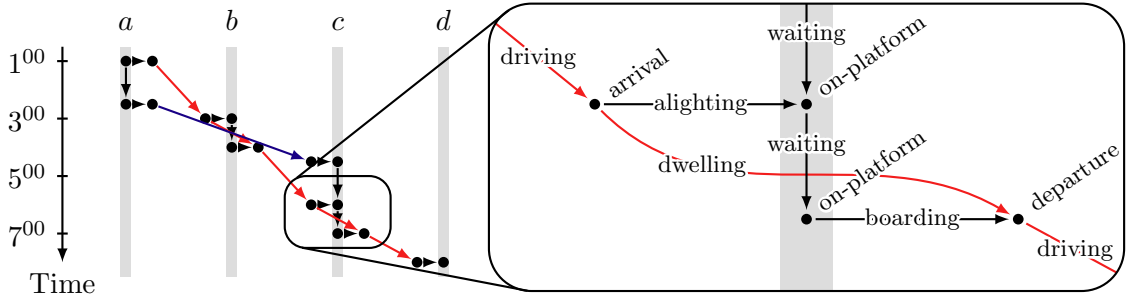

There are five categories of edges connecting these nodes:
For every station $s$, we use \emph{waiting edges} to connect the on-platform nodes of $s$ in a chain $(v^1, \dots, v^k)$ with increasing times, i.e., $\nodetime(v^1) < \dots < \nodetime(v^k)$.
Users may board a vehicle using a \emph{boarding edge} which connects an on-platform node with a departure node of a vehicle~$z$ of common time $\theta$ and station $s$.
Once boarded, the user stays on the vehicle until it arrives at the next station, which is represented by a \emph{driving edge} connecting the departure node with the next station's arrival node of the same vehicle~$z$.
After arriving at a station, the user may alight from the vehicle using an \emph{alighting edge} which connects the arrival node of the vehicle~$z$with the on-platform node of common time $\theta$ and station~$s$.
Unless the vehicle has arrived at its last stop, the user may also choose to stay on the vehicle, which we represent by connecting the arrival node of vehicle $z$with the corresponding departure node at the same station~$s$ using a \emph{dwelling edge}.

For ease of notation, let $E_B$ and $E_D$ denote the set of all boarding and driving edges, respectively.
We denote the traversal time of an edge $e=vw$ by $\tau_e \coloneqq \nodetime(w) - \nodetime(v)$, the traversal time of a $v$-$w$-path $p=(e_1, \dots, e_k)$ by $\tau_p\coloneqq \sum_{e\in p} \tau_e = \nodetime(w) - \nodetime(v)$.
For a driving edge $e \in E_D$ belonging to a vehicle~$z \in Z$, we write $\capacity_e \coloneqq \capacity_z$.
Waiting and driving edges are always time-consuming, dwelling edges may be time-consuming, and boarding and alighting edges are instantaneous.
For a boarding edge $e\in E_B$, we denote the succeeding driving edge by $e^+$.

\Cref{fig:intro-lens} shows a possible generated graph for two vehicles, a \colourtwo one and a \colourone one, and four stations $a$, $b$, $c$\oxc and $d$.
The nodes on the grey rectangles represent the on-platform nodes, the other nodes are the departure and arrival nodes.
As the driving edges together with their associated trips already describe the entire graph, we use a more compact visual representation, encoding trips using a colour scheme.
\Cref{fig:compact} illustrates this compact representation for the graph of \Cref{fig:intro-lens}.

Let us now formalize how the agents of the network are modelled.
We first partition the non-atomic agents into a finite set of groups $J$:
Each group $i\in J$ is assigned an origin station $s_i\in S$ and a destination station $t_i\in S$, a feasible departure time interval $\departureTimes_i \subseteq \mathbb R$ and a target arrival time $\targetTime_i\in \mathbb R$, as well as a total demand $\demand_i$.
Let $\pathsWithoutOutside_i$ denote the set of paths in the time-expanded graph that start during the interval $\departureTimes_i$ and lead from an on-platform node at station $s_i$ to an on-platform node at station $t_i$.
The experienced cost of a path $p\in\pathsWithoutOutside_i$ for agents of group $i$ is then given by
\[
    \pi_{i,p} \coloneqq \beta_i \cdot \tau_p + \gamma^+_i \cdot\max\{0, \patharrival_p -T_i\} + \gamma^-_i\cdot\max\{0, T_i - \patharrival_p\},
\]
where $\patharrival_p$ denotes the arrival time of the path $p$, and $\beta_i$, $\gamma^+_i$\oxc and $\gamma^-_i$ denote the group-specific non-negative penalty factors of the travel time, of late arrival\oxc and of early arrival, respectively.

\newcommand*{\maxwillingness}{\pi_{\mathrm{max}}}
\newcommand*{\elasticDemand}{Q^{\mathrm{el}}}

Some agents may decide to forego using the transit service if their experienced cost would exceed their willingness to travel.
This elastic demand is modelled by a non-increasing function $\elasticDemand_i\colon\R_{\geq0}\to\R_{\geq0}$ that given some cost $\pi$ returns the volume of particles of group $i$ that are willing to travel if the experienced cost does not exceed $\pi$.
We assume $\elasticDemand_i(\maxwillingness)=0$ for some $\maxwillingness\in\R$.
We now subdivide each group $i\in J$ into a finite number of \emph{commodities} $I$ of common willingness to travel:
Let $\{ \pi_{i,1},\dots,\pi_{i,k_i} \} = \{ \pi_{i,p} \mid p\in \pathsWithoutOutside_i \}$ be the set of travel times of all paths $p\in\pathsWithoutOutside_i$ ordered by $\pi_{i,1} < \cdots < \pi_{i,k_i}$.
For each $j\in\{1,\dots,k_i+1\}$, we introduce a commodity $i_j$ consisting of all particles of group $i$ whose willingness to travel is contained in the interval $[ \pi_{i,j-1}, \pi_{i,j} )$ with $\pi_{i,0}\coloneqq 0$ and $\pi_{i,k_i+1}\coloneqq \maxwillingness$. In other words, commodity~$i_j$ contains all particles from group~$i$ that are willing to use the path with cost~$\pi_{i,j-1}$ but not the one with cost~$\pi_{i,j}$.
Thus, commodity~$i_j${} has a demand volume of \[
    \demand_{i_j} \coloneqq \elasticDemand_i(\pi_{i,j-1}) - \elasticDemand_i(\pi_{i,j}),
\]
and we assign it an outside option $\outopt_{i_j}$ with some constant cost $\pi_{\outopt_{i_j}}$ chosen from $(\pi_{i,j-1}, \pi_{i,j})$,
such that the outside option is perceived strictly worse than any path of cost at most $\pi_{i,j-1}$, but strictly better than any other path.
Finally, a particle of commodity $i_j$ can choose a strategy from the set $\pathsWithOutside_{i_j} \coloneqq \pathsWithoutOutside_i \cup \{ \outopt_{i_j} \}$.
\Cref{fig:elastic-demand} demonstrates this classification of the particles into commodities.

In the remainder of the work, we will no longer refer to the groups $J$, but only to the commodities~$I$, where each commodity $i$ is assigned the parameters $s_i$, $t_i$, $\departureTimes_i$, $T_i$, $\beta_i$, $\gamma^+_i$, $\gamma^-_i$, $Q_i$ and $\pi_{\outopt_i}$.
We denote the set of all commodity-path pairs by $\pathsWithOutside\coloneqq \{ (i,p) \mid p\in\pathsWithOutside_i \}$.

\newcommand{\compactrepr}{
 \begin{tikzpicture}[baseline, y=0.5cm, x=0.75cm, thick]
  \newcommand{\height}{8}
  \newcommand{\depxoffset}{0.5}
  \def\yAxisCoord{-1.2}
  \draw[->] (\yAxisCoord,\height - 0.5) -- (\yAxisCoord, -0.5) node[below] {Time};
  \foreach \y in {1,3,5,7}
  \draw (\yAxisCoord + 0.1,\height - \y) -- (\yAxisCoord -0.1,\height - \y) node[left] {\sbTime{\y}{00}};

  \lifelineCompact{0}{$a$}
  \lifelineCompact{1.5}{$b$}
  \lifelineCompact{3}{$c$}
  \lifelineCompact{4.5}{$d$}
    \draw[->, draw=color2] (0, \height - 1) -- (1.5, \height - 3);
  \draw[->, draw=color2] (1.5, \height - 4) -- (3, \height - 6);
  \draw[->, draw=color2] (3, \height - 7) -- (4.5, \height - 8);
  \draw[->, draw=color1] (0, \height - 2.5) -- (3, \height - 4.5);
 \end{tikzpicture}
}

\newcommand{\drawdemandbrace}[3]{%
    \draw [decorate,decoration={brace}, thick] ($(axis cs:0,#1)+(0.25mm,-0.25mm)$) -- ($(axis cs:0,#2)+(0.25mm,0.25mm)$) node [midway, right] {$Q_{i_{#3}}$};%
}
\newcommand{\elasticdemand}{
 \begin{tikzpicture}
  \begin{axis}[
          axis lines = left,
          xlabel = Travel Time,
          ylabel = {Flow volume},
          ymax=1.1,
          xmax=1.1,
          xtick={0,0.3,0.5,0.7,0.9},
          ytick={0},
          xticklabels={$0$,\(\pi_{i,1}\),\(\pi_{i,2}\),\(\pi_{i,3}\),\(\pi_{i,4}\)},
          clip mode=individual,
          scale=0.8
      ]
      \addplot [
          domain=0:1,
          samples=100,
          color=red
      ]
      {-x^2 + 1};

      \foreach \i/\x in {1/0.3,2/0.5,3/0.7,4/0.9} {
              \addplot[mark=none, black, dotted] coordinates {(\x,0) (\x,{1-\x*\x})};
              \addplot[mark=none, black, dotted] coordinates {(\x,{1-\x*\x}) (0,{1-\x*\x})};
          }
      \addlegendentry{\(\elasticDemand_i(\pi)\)}

      \drawdemandbrace{1}{0.91}{1}
      \drawdemandbrace{0.91}{0.75}{2}
      \drawdemandbrace{0.75}{0.51}{3}
      \drawdemandbrace{0.51}{0.19}{4}
      \drawdemandbrace{0.19}{0}{5}
  \end{axis}
 \end{tikzpicture}
}

\begin{figure}[ht]
    \begin{minipage}[c]{0.48\textwidth}
                    \centering
            \compactrepr
            \vspace*{6pt}
            \caption{Compact representation of the time-expanded transit network from \cref{fig:intro-lens}.}
            \label{fig:compact}
            \end{minipage}
    \hfill
    \begin{minipage}[c]{0.48\textwidth}
        \newcommand{\figureCaption}{$\elasticDemand_i(\pi)$ is the volume of particles of group $i$ that are willing to travel at a cost of at most~$\pi$.}
                    \centering
            \elasticdemand
            \caption{\figureCaption}
            \label{fig:elastic-demand}
            \end{minipage}
\end{figure}

\subsection{Side-Constrained User Equilibrium}

A \emph{(path-based) flow $f$} is a vector $(f_{i,p})_{(i, p)\in\pathsWithOutside}$ with $f_{i,p}\in\R_{\geq0}$.
We call the flow~$f$
\begin{itemize}
    \item \emph{demand-feasible},  if $\sum_{p\in\pathsWithOutside_i} f_{i,p} = \demand_i$ holds for all $i\in I$,
    \item \emph{capacity-feasible}, if $f_e\coloneqq \sum_{i,p \in \pathsWithOutside: e\in p} f_{i,p} \leq \nu_e$ holds for all driving edges $e\in E_D$,
    \item \emph{feasible}, if $f$ is both demand- and capacity-feasible.
\end{itemize}
\newcommand*{\capFeasFlows}{\mathcal F^\nu}
\newcommand*{\feasFlows}{\mathcal F^\nu_\demand}
Let $\demFeasFlows$, $\capFeasFlows$, $\feasFlows$ denote the sets of all demand-feasible, capacity-feasible\oxc and feasible flows, respectively.
For a given demand-feasible flow $f$ and two paths $p,q\in\pathsWithOutside_i$ with $f_{i,p}\geq \varepsilon$, we define the \emph{$\varepsilon$-deviation} from $p$ to $q$ %
\[
    f_{i, p\to q}(\varepsilon) \coloneqq f+\varepsilon\cdot(1_{i,q} - 1_{i,p})
\]
as the resulting flow when shifting an $\varepsilon$-amount of flow of commodity $i$ from $p$ to $q$.
We say that $f_{i, p\to q}(\varepsilon)$ is an \emph{admissible} deviation if $(f_{i,p\to q}(\varepsilon))_{e^+} \leq \nu_{e^+}$ holds for all boarding edges $e$ of~$q$.
If $f_{i,p\to q}(\varepsilon)$ is an admissible deviation for some positive $\varepsilon$, then we call $q$ an \emph{available alternative} to $p$ for $i$ given flow $f$.
In other words, $q$ is an available alternative if, after switching some small amount of flow from $p$ to $q$, the path~$q$ does not involve boarding overcrowded vehicles.
Equivalently, all boarding edges $e\in q$ fulfil $f_{e^+} < \nu_{e^+}$ if $e^+ \notin p$, and $f_{e^+}\leq \nu_{e^+}$ if $e^+\in p$.
We denote the set of available alternatives to $p$ for $i$ given $f$ by $\avPaths_{i,p}(f)$.

\begin{definition}\label{def:sc-equilibrium}
    A feasible flow $f$ is a \emph{(side-constrained) user equilibrium} if for all $i\in I$ and $p\in\pathsWithOutside_i$ the following implication holds: \[
        f_{i,p} > 0
        \,\implies\,
        \forall q\in \avPaths_{i,p}(f): \pi_{i,p} \leq \pi_{i,q}.
    \]    For the rest of this work, we use the shorthand \emph{user equilibrium}.
\end{definition}

This means, a feasible flow is a %
user equilibrium if and only if a path is only used if all its better alternative routes are unavailable due to the boarding capacity constraints.

\section{Characterization, Existence\oxc and Price of Stability}\label{sec:characterization-existence-pos}

We characterize user equilibria as defined above in two different ways: as solutions to a quasi-variational inequality and as equilibria in an extended graph with discontinuous cost functions.
In general, the question of what constitutes an equilibrium in a graph with discontinuous cost functions has no clear answer: various equilibrium concepts have been introduced in the literature, which differ in terms of which portions of flow can switch to an alternative path.
However, this distinction is pointless in our setting:
The two extremes where any amount of flow can deviate, as considered by \citet{DafS69}, and where only an infinitesimal small portion may change its path, as defined by \citet{BernsteinS94}, are equivalent in (extended) time-expanded transit networks (\cref{thm:BS-equilibrium iff SC-equilibrium}).
In the following, we refer to the characterization of \citeauthor{BernsteinS94} as \emph{BS-equilibria}.
For these, there is a known existence result under a certain regularity condition.
However, since this condition is not applicable in our setting, we generalize this existence result by showing that a weaker regularity condition is sufficient to guarantee existence of equilibria (\cref{thm:main}).
This allows us to prove that user equilibria exist in schedule-based transit networks with fixed departure times (\cref{thm:existence-metro}).
Finally, we analyse the price of stability.

\subsection{Quasi-Variational Inequalities}

\newcommand*{\tanCone}[2]{T_{#1}({#2})}

Traditional types of user equilibria without hard capacity constraints can be equivalently formulated as a solution to a variational inequality of the form
    \begin{equation}\label{eq:VI}\tag{\ensuremath{\VI(c,D)}}
        \text{Find $f^* \in  D$ such that:} \quad\quad
        \scalar{c(f^*)}{f - f^*} \geq 0 \quad \text{ for all } f \in D,
    \end{equation}
where $D$ is a closed, convex set and $c$ is a continuous cost function.

With the introduction of hard capacity constraints together with boarding priorities, an admissible $\varepsilon$-deviation might lead to capacity violations.
Therefore, such deviations may leave the feasible set $\feasFlows$ and are thus not representable in such a variational inequality, leading us to the concept of quasi-variational inequalities.
We define the set-valued function
\begin{align*}
    \admissibleDeviations\colon \feasFlows \rightrightarrows \R_{\geq0}^{\pathsWithOutside}, \quad
    f \mapsto \{ f_{i,p\to q}(\varepsilon) \mid f_{i,p\to q}(\varepsilon) \text{ is an admissible $\varepsilon$-deviation}, \varepsilon > 0 \}
\end{align*}
that returns for any given flow $f$ the set of all possible flows obtained by any \addmEpsDev\ with respect to $f$.
We now consider the following quasi-variational inequality:
    \begin{equation}\label{eq:QVI-SCDE}\tag{\ensuremath{\QVI}}
        \text{Find }f^* \in  \feasFlows  \text{ such that:} \quad\quad
        \scalar{\pi}{f - f^*} \geq 0 \quad \text{ for all } f\in \admissibleDeviations(f^*).
    \end{equation}

Then, we can characterize user equilibria as follows:

\begin{theorem}\label{thm:QVI}
    A feasible flow $f^*$ is a user equilibrium if and only if it is a solution to the quasi-variational inequality \eqref{eq:QVI-SCDE}.
\end{theorem}
\begin{proof}
    Assume $f^*$ is a user equilibrium, and let $f \coloneqq f^*_{i,p\to q}(\varepsilon)\in \admissibleDeviations(f^*)$ be arbitrary.
    This means $q$ is in $\avPaths_{i,p}(f^*)$ and $\scalar{\pi}{f - f^*} = \varepsilon \cdot (\pi_{i,q} - \pi_{i,p})\geq 0$ holds by the equilibrium condition.
    Thus, $f^*$ solves \eqref{eq:QVI-SCDE}.

    Similarly, if $f^*$ is a solution to \eqref{eq:QVI-SCDE}, we know for all $p\in\pathsWithOutside_i$ and $q\in \avPaths_{i,p}(f)$ that there is some $\varepsilon>0$ such that $f\coloneqq f^*_{i,p\to q}(\varepsilon)\in \admissibleDeviations(f^*)$.
    Therefore, $\pi_{i,q} - \pi_{i,p} = \scalar{\tau}{f - f^*} / \varepsilon \geq 0$, and thus, $f^*$ is a user equilibrium.
\end{proof}

While the existence of solutions to customary variational inequalities in the form of~\eqref{eq:VI} can be shown using Brouwer's fixed point theorem, the existence of solutions to quasi-variational inequalities is not clear upfront.
To establish an existence result, we therefore introduce an alternative characterization of our problem in the next section.

\subsection{Equilibria for Discontinuous Cost Functions}\label{subsec:BS}

\newcommand*{\succE}[1]{{#1}^+}

In this section, we will reformulate the side-constrained user equilibrium as an equilibrium for suitably chosen edge cost functions~$c_{i,e} \colon \demFeasFlows \to \R_{\geq 0}$.
This way we dispense with the explicit side-constraints and instead incorporate them as discontinuities into the cost functions, so that any equilibrium must correspond to a feasible flow.

To translate our path-based cost functions to edge-based functions, we augment the time-expanded transit network $G=(V, E)$ as follows, resulting in the graph $G'=(V', E')$:
For each commodity $i\in I$, we introduce a source node $\alpha_i$ and a sink node $\omega_i$.
For each on-platform node $v\in V$ at the commodity's origin station $s_i$ with $\theta(v)\in \departureTimes_i$, we add an edge $(\alpha_i, v)$ with zero-cost $c_{i, (\alpha_i, v)}\equiv 0$.
This edge represents the departure of particles of commodity $i$ at time $\theta(v)$.
Similarly, for every on-platform node $v\in V$ at the commodity's destination station $t_i$ we add an edge $(v, \omega_i)$ with a cost of $c_{i,(v,\omega_i)}(f) \coloneqq \gamma^+_i \cdot \max\{0, T_i - \theta(v)\} + \gamma^-_i \cdot \max\{0, \theta(v) - T_i\}$, representing the arrival at the destination at time $\theta(v)$.
Finally, for the outside option, we add an edge $(\alpha_i, \omega_i)$ with cost $c_{i,(\alpha_i,\omega_i)} \equiv \pi_{\outopt_i}$.

Note that the strategy set $\pathsWithOutside_i$ of a commodity $i$ corresponds one-to-one to the $\alpha_i$-$\omega_i$-paths in the graph $G'$.
We now define cost functions on the remaining edges such that we can express the cost of a path $p\in\pathsWithOutside_i$ in terms of the costs of its edges (in $E'$) as $c_{i,p}(f)\coloneqq \sum_{e\in p} c_{i,e}(f)$:
The cost of a non-boarding edge $e\in E\setminus E_B$ is given by the time it takes to traverse the edge weighted by $\beta_i$, i.e., $c_{i,e}(f)\coloneqq \beta_i\cdot \tau_e \geq 0$.
Passing a boarding edge takes no time; however, it is only possible to board until the capacity of the vehicle is reached.
We realize this by raising the cost of the boarding edge when the capacity is exceeded to a sufficiently large constant $M$, which is higher than the cost of any available path, e.g., $M\coloneqq \max_{i\in I, p\in\pathsWithOutside_i} \pi_{i,p} + 1$.
This means, for a boarding edge $e\in E_B$, the experienced cost is $c_{i,e}(f)\coloneqq 0$, if $f_{\succE e}\leq \capacity_{\succE e}$, and $c_{i,e}(f)\coloneqq M$, if $f_{\succE e}>\capacity_{\succE e}$.

As a result, the assigned cost of a path $p\in\pathsWithoutOutside_i$ equals
\begin{equation}\label{telescoping-sum}
    c_{i,p}(f) = \pi_{i,p} + \sum_{e\in p\cap E_B} c_{i,e}(f).
\end{equation}

Note that we can define edge cost functions independent of the commodity if all $\beta_i$ are zero (every commodity only cares about arrival time), or all $\beta_i$ are non-zero (no commodity is indifferent about travel time), in which case we can normalize the $\beta_i, \gamma_i^+, \gamma_i^-$, so that $\beta_i = 1$ for every commodity.

The equilibria with respect to these cost functions are exactly the user equilibria in $G$, as the following theorem shows. Here, it does not matter whether coordinated deviations within one commodity are allowed or not. In the formulation, we identify a multi-commodity flow in $G$ with the corresponding flow in $G'$.

\begin{theorem}\label{thm:BS-equilibrium iff SC-equilibrium}
 Let $f$ be a demand-feasible flow. The following statements are equivalent:
 
 \begin{enumerate}[label=(\roman*)]  \item \label{thm:charac-SC-equilibrium} $f$ is a side-constrained user equilibrium;
  \item \label{thm:charac-DS-equilibrium}
   for all $i\in I$, $p\in\pathsWithOutside_i$ with $f_{i,p}>0$, $q\in\pathsWithOutside_i$, and $0 < \varepsilon < f_{i,p}$ it holds that    \[
    c_{i,p}(f) \leq c_{i,q}(f_{i,p\to q}(\varepsilon));
   \]
     \item \label{thm:charac-BS-equilibrium} 
   $f$ is a BS-equilibrium, i.e., for all $i\in I$, $p\in\pathsWithOutside_i$ with $f_{i,p}>0$, and $q \in \pathsWithOutside_i$ it holds that
      \[
    c_{i,p}(f) \leq \liminf_{\varepsilon\downarrow0} c_{i,q}(f_{i,p\to q}(\varepsilon)).
   \]
    \end{enumerate}
\end{theorem}
This means, side-constrained user equilibria in a time-expanded transit network~$G$ can be modelled either as equilibria in the sense of \citet{DafS69} or as equilibria in the sense of \citet{BernsteinS94} in the extended graph~$G'$ with the discontinuous edge cost functions~$c_{i,e}$.

\newcommand{\itemref}[1]{\ref{#1}{}}
\begin{proof}
   \begin{description}
    \item[\qq{\itemref{thm:charac-SC-equilibrium} $\Rightarrow$ \itemref{thm:charac-DS-equilibrium}}:] Let $f$ be a side-constrained user equilibrium, and let $i \in I$, $p \in \pathsWithOutside_i$ with $f_{i,p} > 0$, $q \in \pathsWithOutside_i$\oxc and $0 < \varepsilon < f_{i,p}$. Since $f$ is capacity-feasible, it holds that $c_{i,p}(f) \stackrel{\eqref{telescoping-sum}}= \pi_{i,p} + \sum_{e \in p \cap E_B} c_{i,e}(f) = \pi_{i,p}$.
   
    If $f_{i,p\to q}(\varepsilon)$ is an admissible deviation, then $q \in A_{i,p}(f)$ and we have $c_{i,q}(f_{i,p\to q}(\varepsilon)) \stackrel{\eqref{telescoping-sum}}= \pi_{i,q} + \sum_{e \in q \cap E_B} c_{i,e}(f_{i,p\to q}(\varepsilon)) = \pi_{i,q}$ as well. Then, by the definition of a side-constrained user equilibrium, it follows that $c_{i,p}(f) = \pi_{i,p} \le \pi_{i,q}  = c_{i,q}(f_{i,p\to q}(\varepsilon))$.

    If $f_{i,p\to q}(\varepsilon)$ is not an admissible deviation, then there is a boarding edge~$e$ of $q$ with $(f_{i,p\to q}(\varepsilon))_{e^+} > \nu_{e^+}$. Therefore, $c_{i,q}(f_{i,p\to q}(\varepsilon)) \ge c_{i,e}(f_{i,p \to q}(\varepsilon)) = M$. The definition of $M$ then implies that $c_{i,p}(f) = \pi_{i,p} < M \le c_{i,q}(f_{i,p \to q}(\varepsilon))$.
   \item[\qq{\itemref{thm:charac-DS-equilibrium} $\Rightarrow$ \itemref{thm:charac-BS-equilibrium}}:] Let $i \in I$, $p \in \pathsWithOutside_i$ with $f_{i,p} > 0$\oxc and $q \in \pathsWithoutOutside_i$. If $c_{i,p}(f) \le c_{i,q}(f_{i,p \to q}(\varepsilon))$ for all $\varepsilon > 0$ small enough, then the inequality also holds for the limit inferior.
   \item[\qq{\itemref{thm:charac-BS-equilibrium} $\Rightarrow$ \itemref{thm:charac-SC-equilibrium}}:] Note that $f$ is capacity-feasible:
    Assuming otherwise implies that there is some boarding edge $e\in E_B$ with $f_e > 0$ for which the driving edge $\succE{e}$ is overfilled, i.e., $f_{\succE{e}}>\nu_e$.
    For any path $p$ containing $e$, we have $c_{i,p}(f) \geq M$, which is larger than the cost $\pi_{\outopt_i}$ of the outside option, and therefore it follows $f_{i,p} = 0$ and $f_e = 0$, a contradiction.

    Let $i \in I$, $p \in \pathsWithOutside_i$ with $f_{i,p} > 0$ and $q \in A_{i,p}(f)$. Then $\pi_{i,p} = \pi_{i,p} + \sum_{e \in p \cap E_B} c_{i,e}(f) \stackrel{\eqref{telescoping-sum}}= c_{i,p}(f) \le \liminf_{\varepsilon \downarrow 0} c_{i,q}(f_{i, p \to q}(\varepsilon))$. Since $q$ is an available alternative, there is $\varepsilon^* > 0$ such that $f_{i,p\to q}(\varepsilon^*)$ is an admissible deviation. Then for all $0 < \varepsilon < \varepsilon^*$ it holds that $c_{i,q}(f_{i,p \to q}(\varepsilon)) \stackrel{\eqref{telescoping-sum}}= \pi_{i.q}$. Therefore, $\pi_{i,p} = \liminf_{\varepsilon \downarrow 0} c_{i,q}(f_{i, p \to q}(\varepsilon)) \le \pi_{i,q}$, implying that $f$ is a side-constrained user equilibrium. \qedhere%
  \end{description}%
\end{proof}

\subsection{Fixed Departure Times}

We now introduce the important special case where users care only about their arrival time and are indifferent to whether they depart later or travel for longer.
This scenario arises when people want to travel home after an event with a fixed end time, such as after school, after a plane has landed, or after a concert.
It is precisely in these situations that public transport systems reach their capacity limits and users therefore begin to behave strategically.
We will show later that this assumption is sufficient for the existence of an equilibrium.

Fixed departure times can be retrieved as a special case of our general model by setting either $\beta_i = 0$ (so thattravel time is not considered) or $\departureTimes_i$ to a singleton (so thatany waiting time at the start station is counted as travel time).
The following \lcnamecref{lem:no-departure-time-choice} shows that these two representations are indeedequivalent.
For a path $p \in \pathsWithOutside$ that starts at an on-platform node~$v$ of station~$z$ at time $\theta(v) \in \Theta_i$, let $\bar p$ be the path that starts at the earliest on-platform node $w$ of $z$ with $\theta(w)\in \Theta_i$, uses waiting edges until time~$\theta(v)$, and then continues as $p$.
Let $\theta_i'\coloneqq \theta(w)$, and let $\bar f$ be the flow obtained by rerouting each path flow $f_{i,p}>0$ to the extended path~$\bar p$ for all $i \in I$.

\begin{lemma}\label{lem:no-departure-time-choice}
 Let $\beta_i, \gamma_i^+, \gamma_i^-$ as well as demands~$Q_i$, target arrival times~$T_i$\oxc and feasible departure time intervals~$\departureTimes_i$ be given for all $i \in I$. For $(i,p) \in \mathcal P$ let
 \begin{align*}
  \pi_{i,p} &\coloneqq \beta_i \cdot \tau_p + \gamma_i^+ \cdot \max\{0,\, \patharrival_p -T_i\} + \gamma^-_i\cdot\max\{0,\, T_i - \patharrival_p\},\\
  \pi_{i,p}' &\coloneqq (\gamma_i^+ + \beta_i) \cdot \max\{0,\, \patharrival_p -T_i\} + (\gamma^-_i - \beta_i) \cdot\max\{0,\, T_i - \patharrival_p\}.
 \end{align*}
 For a demand-feasible flow~$f$ for the departure time intervals~$\Theta_i$, the following are equivalent:
 \begin{enumerate}[label=(\roman*)]  \item $f$ is a user equilibrium for the departure time intervals~$\Theta_i$ and path costs~$\pi_{i,p}'$. \label{no-traveltime-costs}
  \item $\bar f$ is a user equilibrium for the departure times $\Theta_i' \coloneqq \{\theta_i'\}$ (i.e., each $\mathcal P_i^\circ$ contains only paths starting at time~$\theta_i'$) and the path costs~$\pi_{i,p}'$. \label{singleton-without-traveltime-costs}
  \item $\bar f$ is a user equilibrium for the departure times $\Theta_i' \coloneqq \{\theta_i'\}$ and the path costs~$\pi_{i,p}$. \label{singleton}
 \end{enumerate}
\end{lemma}
\begin{proof}
 Consider an arbitrary commodity~$i$, and let \optDisplay{\pathsWithOutside_i' \coloneqq \bigl\{p \in \mathcal P_i \bigm| p = \outopt_i \text{ or } p \text{ is a } v\text{-}w\text{-path with } \theta(v) = \theta_i'\bigr\}.}
    \begin{description}
   \item[\qq{\itemref{no-traveltime-costs} $\Leftrightarrow$ \itemref{singleton-without-traveltime-costs}}:]
  For every path~$p \in \pathsWithOutside_i$ it holds that $\bar p \in \mathcal \pathsWithOutside'$ and $\pi_{i,p}' = \pi_{i,\bar p}'$ because the cost only depends on the arrival time. This implies that there are no improving available alternative paths for $f$ if and only this holds for $\bar f$.
  \item[\qq{\itemref{singleton-without-traveltime-costs} $\Leftrightarrow$ \itemref{singleton}}:] For every path~$p \in \pathsWithOutside_i'$ it holds that
     \begin{align*}
    \pi_{i,p}' &= (\gamma_i^+ + \beta_i) \cdot \max\{0,\,\patharrival_p - \targetTime_i\} + (\gamma_i^- - \beta_i) \cdot \max\{0,\,\targetTime_i - \patharrival_p\} \\
    &= \beta_i \cdot (\patharrival_p - \targetTime_i) + \gamma_i^+ \cdot \max\{0,\,\patharrival_p - \targetTime_i\} + \gamma_i^- \cdot \max\{0,\,\targetTime_i - \patharrival_p\} \\
    &= \beta_i \cdot (\patharrival_p - \theta_i') + \gamma_i^+ \cdot \max\{0,\,\patharrival_p - \targetTime_i\} + \gamma_i^- \cdot \max\{0,\,\targetTime_i - \patharrival_p\} - \beta_i \cdot (\targetTime_i - \theta_i') \\
    &= \pi_{i,p} - \beta_i \cdot (\targetTime_i - \theta_i'),
   \end{align*}
      which implies that the costs of corresponding paths differ only by a commodity-specific constant.
   Consequently, a path is an improving alternative for $\pi'$ if and only if it is one for $\pi$.\qedhere%
  \end{description}%
\end{proof}

\begin{definition}
 We say that a given instance has \emph{fixed departure times} (FDT) if $\Theta_i$ is a singleton for every commodity $i\in I$.
\end{definition}

From \cref{lem:no-departure-time-choice} it follows that, when considering fixed departure times, we may assume for each commodity~$i$ both that $\Theta_i$ is a singleton and that $\beta_i = 0$.
Therefore, the cost of a path depends only on the arrival time, so it can be written as $\pi_{i,p} = \pi_i'(\patharrival_p)$, where
\begin{equation}
 \pi_i'(t) \coloneqq \gamma_i^+ \cdot \max\{0,\, t - T_i\} + \gamma_i^- \cdot \max\{0,\, T_i - t\}. \label{eq:pi-prime}
\end{equation}
Note that since each commodity has a separate destination node~$\omega_i$, the edges leading to $\omega_i$ are only used by commodity~$i$.
Therefore, the functions $c_{i,e}$ are independent of the commodity:
\[
    c_e(f) = \begin{cases}  c_{i,e}(f) &\text{ if } e = (v, \omega_i) \text{ for some } i \in I, \\
    M &\text{ if } e \text{ is a boarding edge with } f_{e^+} \ge \capacity_{e^+}, \\
    0 &\text{ else.} \end{cases}
\]
The case when $T_i = 0$, $\gamma_i^+ > 0$ means that commodity~$i$ simply aims to minimize its arrival time, i.e., $\pi_i'(t) = \gamma_i^+ t$.

\subsection{Generalization of Bernstein and Smith's Existence Result}

In preparation for the analysis of existence of user equilibria in schedule-based transit networks, we revise the theory developed by \Citet{BernsteinS94} for general networks in this subsection.
They proved the existence of BS-equilibria in the case that each path cost function has the form $c_{i,p} = \sum_{e \in p} c_e$, where $c_e \colon \demFeasFlows \to \R_{\ge 0}$, $e \in E$, are lower-semicontinuous, bounded functions that satisfy the following regularity condition.

\begin{definition}\label{def:regular-cost-function}
    A cost function $c \colon \demFeasFlows \to \R_{\geq 0}^E$ is \emph{regular} if it satisfies
    \[
        \liminf_{\varepsilon\downarrow0} c_{i,q}(f_{i,p\to q}(\varepsilon))
        =
        \sum_{e\in p\cap q} c_e(f) + \sum_{e\in q\setminus p} \bar c_e(f)
    \]
    for all $f\in\demFeasFlows$, $i\in I$\oxc and paths $p,q\in\pathsWithOutside_i$ with $f_{i,p} > 0$,
    where $\bar c_e$ is the upper hull of $c_e$ defined as  \optDisplay{
        \bar c_e(f) \coloneqq \lim_{\varepsilon\downarrow 0} \sup\{ c_e(x) \mid x\in \demFeasFlows, \norm{x-f} < \varepsilon \}.
    }
\end{definition}

\begin{remark}   
    The cost function defined in \Cref{subsec:BS} for schedule-based transit networks is not regular even for fixed departure times.
    This is illustrated by the network in \Cref{fig:intro-lens}:
    Assume there is a single commodity with origin $a$, destination $d$ and demand $2$ minimizing its arrival time (i.e., $T_i = 0, \gamma_i^+ = 1$),
    and assume that both vehicles have capacity $1$.
    Let $p$ be the $a$-$d$-path using only the \colourtwo vehicle, and let $q$ be the $a$-$d$-path using both vehicles.
    Let $f$ be the flow sending one unit along $p$ and the remaining unit along the commodity's outside option $\outopt_i$.
    Then, $\bar c_e(f) = M$ holds for the boarding edge $e$ of the \colourtwo vehicle at station $c$ (as $\demFeasFlows$ contains $f_{i,\outopt_i\to p}(\varepsilon)$ for $\varepsilon \leq 1$).
    This implies
    \[
        \liminf_{\varepsilon\downarrow0} c_{q}(f_{i,p\to q}(\varepsilon))
        = \tau_q < M \leq \sum_{e\in p\cap q} c_e(f) + \sum_{e\in q\setminus p} \bar c_e(f).
    \]
    On the left-hand side it is noticed that the flow on the last driving edge is unchanged and boarding remains possible, whereas the right-hand side is oblivious to the flow reduction along $p$.
\end{remark}

The goal of this subsection is to show that the following weaker regularity condition is actually sufficient for existence.
Here, we also allow edge costs to be commodity-specific, i.e., the path costs have the form $c_{i,p} = \sum_{e \in p} c_{i,e}$ for lower-semicontinuous functions~$c_{i,e}$.

\begin{definition}
    A cost function $c \colon \demFeasFlows \to \R^{I\times E}_{\ge 0}$ %
    is called \emph{weakly regular} if
            the following implication holds for all demand-feasible flows $f\in\demFeasFlows$, $i\in I$\oxc and $p\in\pathsWithOutside_i$ with $f_{i,p} > 0$: \[
        c_{i,p}(f) \leq \min_{q\in\pathsWithOutside_i} \sum_{e\in p\cap q} c_{i,e}(f) + \sum_{e\in q\setminus p} \bar c_{i,e}(f)
        \,\implies\,
        c_{i,p}(f) \leq \min_{q\in\pathsWithOutside_i} \liminf_{\varepsilon \downarrow 0} c_{i,q}(f_{i,p\to q}(\varepsilon)).
        \]
    \end{definition}

Note that a regular cost structure $c\colon\feasFlows \to \R^E_{\geq0}$ can be interpreted as a function $c'\colon \feasFlows \to \R^{I\times E}_{\geq0}$ by defining $c'_{i,e} \coloneqq c_e$ for all $i\in I$, $e\in E$.
\begin{proposition}
    A regular cost structure~$c$ is also weakly regular.
\end{proposition}
\begin{proof}
    Assume, the left side of the implication in the definition of weak regularity holds true for some $p\in\pathsWithOutside_i$, $f\in\demFeasFlows$, and let $q \in \pathsWithOutside_i$.
    Then, it also holds
    \iftrue
    \[
        c_{i,p}(f)
        \leq \sum_{e\in p\cap q} c_{i,e}(f) + \sum_{p \setminus q} \bar c_{i,e}(f)
        = \liminf_{\varepsilon\downarrow 0} c_{i,q}(f_{i,p\to q}(\varepsilon)),
    \]
    \else
    \begin{align*}
        c_{i,p}(f)
        &\leq \sum_{e\in p\cap q} c_{i,e}(f) + \sum_{p \setminus q} \bar c_{i,e}(f)
        \\
        &= \liminf_{\varepsilon\downarrow 0} c_{i,q}(f_{i,p\to q}(\varepsilon)),
    \end{align*}
    \fi
    where we apply regularity for the last equation.
    Taking the minimum over all $q\in\pathsWithOutside_i$ yields weak regularity.
\end{proof}

\begin{theorem}\label{thm:main}
    If $c \colon \demFeasFlows \to \R_{\geq 0}^{I\times E}$ is a lower-semicontinuous, bounded\oxc and weakly regular cost structure, a BS-equilibrium exists.
\end{theorem}

The proof follows the same ideas as \cite[Theorem~2]{BernsteinS94}.

\begin{proof}
    Let $M$ be a common upper bound for all functions $c_{i,e}$, $i\in I$, $e \in E$.
    There exists for each pair $(i,e)$ a sequence of continuous functions $c^{(n)}_{i,e} \colon \demFeasFlows\to [0,M]$ such that $c^{(n)}_{i,e}(f)\uparrow c_{i,e}(f)$ holds for all $f \in \demFeasFlows$.
    For each $n\in\N$, there is a Wardrop equilibrium $f^{(n)}\in\demFeasFlows$ w.r.t.\ the path cost function~$(c^{(n)}_{i,p})_{(i,p) \in \pathsWithOutside}$ defined by $c_{i,p}^{(n)}(f) \coloneqq \sum_{e \in p} c_{i,e}^{(n)}(f)$~\citep[cf.][]{Smith79}. That means \[
        f^{(n)}_{i,p}>0 \implies c^{(n)}_{i,p}(f^{(n)})\leq  c^{(n)}_{i,q}(f^{(n)})
    \]
    holds for all $i\in I$ and paths $p,q\in\pathsWithOutside_i$.
        Equivalently, we have \begin{equation}\label{eq:sequence-eq-cond}
        f^{(n)}_{i,p}>0 \implies \sum_{e\in p\setminus q} c^{(n)}_{i,e}(f^{(n)})\leq \sum_{e\in q\setminus p} c^{(n)}_{i,e}(f^{(n)}).
    \end{equation}
    The sequence $(f^{(n)}, c^{(n)}(f^{(n)}))$ is contained in the compact set $\demFeasFlows\times [0,M]^{I\times E}$ and therefore has a convergent sub-sequence with some limit $(f, x)$; we pass to this sub-sequence.

    By the upper-semicontinuity of the upper hull and the monotonicity of the sequence of cost functions, we have for all $e \in E$ and $i\in I$
    \begin{equation}\label{eq:x_e upper bound}
        \bar c_{i,e}(f)
        \geq \limsup_{n\to\infty} \bar c_{i,e}(f^{(n)})
        \ge \limsup_{n \to \infty} c_{i,e}(f^{(n)})
        \geq \lim_{n\to\infty} c^{(n)}_{i,e}(f^{(n)})
        = x_{i,e}.
    \end{equation}

    Let $\lambda > 0$.
    First, since $(c^{(n)}_{i,e})_n$ converges pointwise to $c_{i,e}$, there exist $n_0\in\N$ such that $c^{(n_0)}_{i,e}(f) \ge c_{i,e}(f) - \lambda/2$.
    Second, since $c_{i,e}^{(n_0)}$ is continuous, there is $\delta > 0$ such that for all $g \in \demFeasFlows$ with $\norm{f-g} < \delta$ we have $c_{i,e}^{(n_0)}(g) \ge c^{(n_0)}_{i,e}(f) - \lambda/2$.
    As $(c_{i,e}^{(n)})_n$ is a pointwise increasing sequence, we then have for all $n \ge n_0$ that $c_{i,e}^{(n)}(g) \ge c_{i,e}(f) - \lambda$.
    Third, since $(f^{(n)})_n$ converges to $f$, there is $n_1$ such that $\smallnorm{f-f^{(n)}}<\delta$ holds for all $n\geq n_1$.
    In conclusion, $c^{(n)}_{i,e}(f^{(n)}) \geq c_{i,e}(f)-\lambda$ holds for $n \ge \max\{n_0,n_1\}$.
    Since $\lambda > 0$ was arbitrary, we deduce \begin{equation}\label{eq:x_e lower bound}
        x_{i,e} = \lim_{n\to\infty}c^{(n)}_{i,e}(f^{(n)}) \geq c_{i,e}(f).
    \end{equation}
        
    Let $i\in I$ and $p\in \pathsWithOutside_i$ with $f_{i,p}>0$.
    There exists $n_0\in\N$ with $f^{(n)}_{i,p}>0$ for all $n\geq n_0$.
    Let $q\in\pathsWithOutside_i$ be an arbitrary other path.
    Taking the limit of~\eqref{eq:sequence-eq-cond} yields \(
    \sum_{e\in p\setminus q} x_{i,e} \leq \sum_{e\in q\setminus p}  x_{i,e}
    \), and by applying the inequalities \eqref{eq:x_e upper bound} and \eqref{eq:x_e lower bound}, we get \optDisplay{
        \sum_{e\in p\setminus q} c_{i,e}(f) \leq \sum_{e\in q\setminus p} \bar c_{i,e}(f).
    }
    Adding $c_{i,e}(f)$ for each $e\in p\cap q$ to both sidesand taking the minimum over all $q\in\pathsWithOutside_i$, this shows \optDisplay{
        c_{i,p}(f) \leq \min_{q\in\pathsWithOutside_i} \sum_{e\in p\cap q} c_{i,e}(f) + \sum_{e\in q\setminus p} \bar c_{i,e}(f).
    }
    Thus, we can apply weak regularity, which implies that $q$ is not an improving alternative path. Since this holds for all $q \in \pathsWithoutOutside_i$ and all $p \in \pathsWithoutOutside_i$ with $f_{i,p} > 0$, the flow~$f$ is a BS-equilibrium.
\end{proof}

\subsection{Existence of Equilibria in Schedule-Based Transit Networks}\label{subsec:Existence}

In general, the existence of a user equilibrium with departure choice is not guaranteed, as was already discovered by \citet{NguyenPM01}; this is illustrated in \Cref{example:non-existence-dtc}.
In fact, deciding whether a user equilibrium exists is actually an NP-hard problem, as we will show in \cref{subsec:NP-hardness}.
On the other hand, in \cref{subsec:finite-time algorithm}, we will provide an exponential-time algorithm that decides this question and computes an equilibrium if it exists.
For fixed-departure time choice, however, we can derive existence using the generalized theorem from the previous section (see~\Cref{thm:existence-metro}).
Before diving into this proof, we first show that in the general case, user equilibria do not necessarily exist:

\begin{example}\label{example:non-existence-dtc}
    We consider the network shown in \cref{fig:non-existence-dtc}:
    All vehicles have capacity $1$, there is a single commodity with a demand of $2$, and we have $\beta = 1$ and $\gamma^+=\gamma^-=0$.
    Assume there is a user equilibrium~$f$.
    There are three reasonable paths: the path $p_1$ of minimal cost starts late and takes the direct \colourone{}-vehicle edge from $s$ to $t$, the second-best path $p_2$ starts early and uses the \colourone{} vehicle including the detour via $v$, and the worst path $p_3$ starts early, takes the \colourtwo{} vehicle and arrives later than $p_1$ and $p_2$.
    As the capacity of both vehicles is $1$, the path $p_3$ must be used by a flow volume of $1$.
    Furthermore, for particles using $p_3$ to fulfil the equilibrium condition, path $p_2$ must be an unavailable alternative to $p_3$; this implies that the first driving edge of the \colourone{} vehicle from $s$ to $v$ must already be occupied.
    Therefore, path $p_2$ must be used by a flow volume of $1$ as well.
    The particles on path $p_2$, however, perceive path $p_1$ as a better available alternative, which means that the equilibrium condition cannot be satisfied.
\end{example}

\begin{remark}
    It is worth noting that the above example also shows that $\varepsilon$-approximate user equilibria do not generally exist, even for arbitrarily large $\varepsilon$.
    Here, a feasible flow $f$ is called an \emph{$\varepsilon$-approximate user equilibrium} if it fulfils $\pi_{i,p}\leq(1+\varepsilon)\cdot \pi_{i,q}$ for all $(i,p)\in\pathsWithOutside$ with $f_{i,p}>0$ and $q\in A_{i,p}(f)$.
    To acknowledge this, let us delay the arrival of the \colourone{} vehicle at the second stop at $s$ and at the final stop at $t$ such that $\pi_{p_2} > (1+\varepsilon)\pi_{p_1}$ holds.
    Further, we delay the arrival of the \colourtwo{} vehicle to obtain $\pi_{p_3} > (1+\varepsilon)\pi_{p_2}$.
    As above, an $\varepsilon$-approximate user equilibrium would assign $p_3$ a flow volume of~$1$; by the equilibrium condition, path $p_2$ must be unavailable implying that $p_2$ is used by a flow volume of $1$.
    For these particles, however, path $p_1$ is an available alternative whose cost is smaller than $\pi_{p_2}/(1+\varepsilon)$, contradicting the approximate equilibrium condition.
\end{remark}

\newcommand{\nonexistence}{
 \begin{tikzpicture}[yscale=.5, thick]
  \newcommand\height{5.5}
  \newcommand\depxoffset{0.5}

  \def\yAxisCoord{-1.5}
  \draw[->] (\yAxisCoord,\height) -- (\yAxisCoord, -0.5) node[below] {Time};
  \foreach \y in {1,3,5}
  \draw (\yAxisCoord + 0.1,\height - \y) -- (\yAxisCoord -0.1,\height - \y) node[left] {\sbTime{\y}{00}};

  \lifelineCompact{0}{$s$}
  \lifelineCompact{2}{$v$}
  \lifelineCompact{4}{$t$}

  \draw[->, draw=color1] (0, \height - 1) -- (2, \height - 2);
  \draw[->, draw=color1] (2, \height - 2) -- (0, \height - 3);
  \draw[->, draw=color1] (0, \height - 3) -- (4, \height - 4);
  \draw[->, draw=color2] (0, \height - 1) -- (4, \height - 5);
 \end{tikzpicture}
}

\begin{figure}[t]
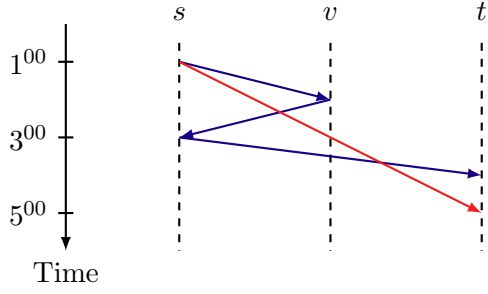

      \centering
     \nonexistence
     \caption{An example for non-existence when considering departure time choice.}
     \label{fig:non-existence-dtc}
 \end{figure}

\medskip
We now turn to the case of fixed departure times.

\begin{theorem}\label{thm:existence-metro}
  In schedule-based transit networks with fixed departure times, a user equilibrium always exists.
\end{theorem}

We want to apply \cref{thm:main} to establish the existence of user equilibria for fixed departure times.
While it is clear that the cost functions are lower-semicontinuous and bounded, some effort is required to show that they fulfil weak regularity.
The idea is that given a path $p$ and a path $q$ minimizing $\liminf_{\varepsilon \downarrow 0} c_{i,q}(f_{i,p\to q}(\varepsilon))$ we consider the last common node $v$ of $p$ and $q$, and define $q'$ as the path formed by concatenating the prefix of $p$ up until $v$ with the suffix of $q$ starting from $v$.
For boarding edges $e$ on the second part of $q'$, we can then show $\liminf_{\varepsilon\downarrow0}c_e(f_{i,p\to q'}(\varepsilon)) = \bar c_e(f)$, while boarding edges $e$ on the first part fulfil $\liminf_{\varepsilon\downarrow 0}c_e(f_{i,p\to q'}(\varepsilon)) = c_e(f)$.
Observing that $\liminf_{\varepsilon\downarrow0}c_{i, q'}(f_{i,p\to q'}(\varepsilon)) \leq \liminf_{\varepsilon\downarrow0}c_{i,q}(f_{i,p\to q}(\varepsilon))$ then concludes the argument.

\begin{proposition}\label{lem:cost on common edge}
    Let $f$ be a flow, let $i \in I$, let $p, q \in \pathsWithOutside_i$, let $e \in p \cap q \cap E_B$, and let $\varepsilon \in [0, f_{i,p}]$. Then $c_{i,e}(f_{i,p\to q}(\varepsilon)) = c_{i,e}(f)$.
\end{proposition}
\begin{proof}
    The function $c_e$ depends only on the flow value on $e^+$. Since $e^+$ also lies in $p \cap q$, this flow value is the same in $f$ and in $f_{i,p\to q}(\varepsilon)$.
\end{proof}

\begin{notation}
    For a path $p$ and a node $v$ occurring in $p$, we denote the prefix of $p$ up to this node by $p_{\leq v}$ and the suffix of $p$ starting from $v$ by $p_{\geq v}$.
\end{notation}

\begin{lemma}\label{lem:prefix-indepence}
    Assume fixed departure times, and let $f$ be a flow, let $i \in I$, and let $p,q\in \pathsWithOutside_i$ with $c_{i,p}(f) < M$.
    Let $v$ be the last common node of $p$ and $q$, and let $q'$ be the concatenation of $p_{\leq v}$ and $q_{\geq v}$.
    Then it holds
    \[
        \liminf_{\varepsilon \downarrow 0} c_{i,q'}(f_{i,p\to q'}(\varepsilon)) \le \liminf_{\varepsilon \downarrow 0} c_{i,q}(f_{i,p\to q}(\varepsilon)).
    \]
\end{lemma}
\begin{proof}
    Let $\varepsilon > 0$.
    The cost of the subpaths $p_{\leq v}$ and $q_{\leq v}$, neglecting boarding edges, are both zero.
    Since the suffixes of $q$ and $q'$ coincide, the expressions $c_{i,q}(f_{i,p\to q}(\varepsilon))$ and $c_{i,q'}(f_{i,p\to q'}(\varepsilon))$ differ only in the cost of their corresponding boarding edges.
    Applying \cref{lem:cost on common edge} to the flow $f_{i,p\to q}(\varepsilon)$ and paths $q$ and $q'$, we see that the costs of the edges in $q$ after node~$v$ are equal under $f_{i,p\to q}(\varepsilon)$ and $f_{i,p\to q'}(\varepsilon)$.
    Moreover, by applying the \lcnamecref{lem:cost on common edge} to $f$ and the paths $p$ and $q'$, we conclude that $c_{e}(f_{i,p\to q'}(\varepsilon)) = c_{e}(f)$ for all $e \in p_{\leq v} = q'_{\leq v}$.
    Therefore, we have
    \begin{align*}
        c_{i,q}(f_{i,p\to q}(\varepsilon)) - c_{i,q'}(f_{i,p\to q'}(\varepsilon))
         & = \sum_{e \in q_{\leq v} \cap E_B} c_{e}(f_{i,p\to q}(\varepsilon)) - \sum_{e \in p_{\leq v} \cap E_B} c_{e}(f) \\
         & = \sum_{e \in q_{\leq v} \cap E_B} c_{e}(f_{i,p\to q}(\varepsilon)) \geq 0,
    \end{align*}
    where the last equation holds because of the assumption $c_{i,p}(f) < M$.
\end{proof}

\begin{proof}[Proof of \Cref{thm:existence-metro}]
    We show that the cost structure $c$ defined in \cref{subsec:BS} fulfils the conditions of \cref{thm:main}.
    This implies the existence of a user equilibrium as per \cref{thm:BS-equilibrium iff SC-equilibrium}.

    Clearly, $c_{i,e}$ is bounded and lower semi-continuous for all $i\in I$, $e \in E$.
    To show weak regularity, let $f\in\demFeasFlows$, $i\in I$\oxc and $p\in \pathsWithOutside_i$ with $f_{i,p} > 0${} fulfil  \[
        c_{i,p}(f) \leq \min_{q'\in\pathsWithOutside_i} \sum_{e\in p\cap q'} c_{e}(f) + \sum_{e\in q'\setminus p} \bar c_{e}(f).
    \]
    Applying the above equation for the outside option $q'=\outopt_i$ results in $c_{i,p}(f) \leq \pi_{\outopt_i} < M$.

    Let $q$ be an arbitrary path in $\pathsWithOutside_i$, let $v$ be the last common node of $p$ and $q$, and let $q'$ be the concatenation of $p_{\leq v}$ and $q_{\geq v}$.
    By \Cref{lem:prefix-indepence}, we have
    \[
        \liminf_{\varepsilon \downarrow 0} c_{i,q'}(f_{i,p\to q'}(\varepsilon)) \le \liminf_{\varepsilon \downarrow 0} c_{i,q}(f_{i,p\to q}(\varepsilon)),
    \]
    and it suffices to show that
        \[
        \sum_{e\in q'\cap p} c_e(f) + \sum_{e\in q'\setminus p} \bar c_e(f)
        \leq \liminf_{\varepsilon\downarrow0} c_{i,q'}(f_{i,p\to q'}(\varepsilon))
        = \sum_{e \in q'} \liminf_{\varepsilon \downarrow 0} c_e(f_{i,p \to q'}(\varepsilon)).
    \]
        As the cost of non-boarding edges is constant, we can restrict our analysis to boarding edges.
    Note that there is no boarding edge $e\in E_B$ on the path $q_{\geq v}$ for which $\succE e$ is also on $p$; otherwise $p$ and $q$ would not be disjoint after $v$.
    This means, for any boarding edge $e$ on the path $q_{\geq v}$, we have
    \[
        \liminf_{\varepsilon\downarrow0} c_{i,e}(f_{i,p\to q'}(\varepsilon))
        = \liminf_{\varepsilon\downarrow0} c_{i,e}(f+\varepsilon \cdot 1_{q'})
        = \left\{\begin{aligned}
            M, & \text{ if $f_{\succE e} \geq \nu_{\succE e}$,} \\
            0, & \text{ otherwise}
          \end{aligned}\right\}
        = \bar c_{i,e}(f).
    \]
    For a boarding edge $e$ on the subpath $p_{\leq v}$, \cref{lem:cost on common edge} implies
    $\liminf_{\varepsilon\downarrow 0} c_{i,e}(f_{i,p\to q'}(\varepsilon)) = c_{i,e}(f)$%
    , which concludes the proof.
\end{proof}

\subsection{Price of Stability}

The existence of user equilibria in the case of fixed departure times allows us to study the quality of these equilibria.
Well-studied measures of quality include the price of anarchy and the price of stability.
These compare the social cost of a user equilibrium to the system-optimal flow.
To this end, the \emph{social cost} of a flow $f$ is defined as the sum of the costs of all paths weighted by the flow on them,~i.e., \[
    \pi(f)\coloneqq \sum_{(i,p)\in\pathsWithOutside} f_{i,p} \cdot \pi_{i,p}.
\]
We call a feasible flow \emph{system-optimal} if it minimizes the social cost among all feasible flows.

For a given problem instance $I$ consisting of a network and a set of commodities, the \emph{price of stability} is defined as the social-cost ratio of the \emph{best} user equilibrium and a system optimum,~i.e., \[
    \mathrm{PoS}(I) \coloneqq \frac{\inf_{f\in \mathrm{EQ}(I)} \pi(f)}{\min_{f \in \feasFlows(I)} \pi(f)},
\]
where $\mathrm{EQ}(I)$ is the set of user equilibria for $I$ and $\feasFlows(I)$ is the set of feasible flows for instance $I$.
The price of stability is a lower bound on the so-called \emph{price of anarchy}, which is the social-cost ratio of the \emph{worst} user equilibrium compared to the system optimum.

Note that the price of stability and the price of anarchy differ only if the social cost is not unique across all user equilibria of a fixed instance.
In fact, such instances exist even for single-commodity networks with fixed departure times, as the following example shows:

\begin{example}\label{ex:non-uniqueness}
    Consider the network in \Cref{fig:non-uniqueness}.
    We assume that the single commodity has a demand of~$2$ with fixed departure time $\departureTimes = \{ \sbTime{1}{00} \}$ and parameters $\beta=1$ and $\gamma^+ = \gamma^- = 0$,
    and that all vehicles have a capacity of~$1$.
    Let $p_1$ denote the \colourthree{} path, $p_2$ the \colourone-\colourone-\colourthree{} path, $p_3$ the \colourone-\colourtwo{} path\oxc and $p_4$ the \colourfour{} path.
    Clearly, $\tau_{p_1} = \tau_{p_2} < \tau_{p_3} < \tau_{p_4}$.
    In a user equilibrium, the \colourthree{} edge is fully utilized as it leads to the earliest arrival at~$t$; the remaining particles try to use either~$p_3$ or, as a last resort,~$p_4$.
    More specifically, for every $\lambda\in[0,1]$, we can define a user equilibrium~$f^\lambda$ with $f^\lambda_{p_1} \coloneqq f^\lambda_{p_3} \coloneqq \lambda$ and $f^\lambda_{p_2} \coloneqq f^\lambda_{p_4} \coloneqq 1 - \lambda$.
    The social cost of~$f^\lambda$ can be computed as $\pi(f)\coloneqq \tau_{p_1} + \lambda \tau_{p_3} + (1-\lambda) \tau_{p_4}$.
    In particular, the user equilibrium~$f^1$ has strictly smaller social cost than the user equilibrium~$f^0$.
\end{example}

\newcommand{\nonuniqueness}{
 \begin{tikzpicture}[yscale=0.5, thick]
  \newcommand{\height}{6}
  \lifelineCompact{0}{$s$}
  \lifelineCompact{2}{$v$}
  \lifelineCompact{4}{$t$}

  \def\yAxisCoord{-1.5}
  \draw[->] (\yAxisCoord,\height - .5) -- (\yAxisCoord, -0.5) node[below] {Time};
  \foreach \y in {1,3,5}
  \draw (\yAxisCoord + 0.1,\height - \y) -- (\yAxisCoord -0.1,\height - \y) node[left] {\sbTime{\y}{00}};

  \draw[->, draw=color1] (0, \height - 1) -- (2, \height - 1.5);
  \draw[->, draw=color1] (2, \height - 1.5) -- (0, \height - 2);
  \draw[->, draw=color2] (2, \height - 2) -- (4, \height - 4.5);
  \draw[->, draw=color3] (0, \height - 2.5) -- (4, \height - 3.5);
  \draw[->, draw=color4] (0, \height - 4.5) -- (4, \height - 5.5);
 \end{tikzpicture}
}

In the following, we give an example that shows that both the price of anarchy and the price of stability are unbounded even if we restrict to single-commodity and fixed-departure-time instances.

\begin{proposition}
    Even for single-commodity instances with fixed departure time choice, the price of stability is unbounded.
\end{proposition}

\newcommand{\unboundedpoa}{
 \begin{tikzpicture}[yscale=0.5, thick]
  \newcommand{\height}{6}
  \lifelineCompact{0}{$s$}
  \lifelineCompact{2}{$v$}
  \lifelineCompact{4}{$t$}

  \def\yAxisCoord{-1.5}
  \draw[->] (\yAxisCoord,\height - .5) -- (\yAxisCoord, -0.5) node[below] {Time};
  \foreach \y in {1,3,5}
  \draw (\yAxisCoord + 0.1,\height - \y) -- (\yAxisCoord -0.1,\height - \y) node[left] {\sbTime{\y}{00}};

    \draw[->, draw=color1] (0, \height - 1) -- (2, \height - 2);
  \draw[->, draw=color1] (2, \height - 2) -- (4, \height - 5);
  \draw[->, draw=color4] (2, \height - 2) -- (0, \height - 3);
  \draw[->, draw=color4] (0, \height - 3) -- (4, \height - 4);
  \draw[->, draw=color2] (0, \height - 5) -- (4, \height - 6);
 \end{tikzpicture}
}

\begin{figure}[t]
    \begin{minipage}{0.45\textwidth}
        \newcommand{\figureCaption}{A network with two user equilibria with different social costs.}
                    \centering
            \nonuniqueness
            \caption{\figureCaption}
            \label{fig:non-uniqueness}
            \end{minipage}\hfill%
    \begin{minipage}{0.45\textwidth}
                    \centering
            \unboundedpoa
            \caption{A network illustrating unboundedness of the price of stability.}
            \label{fig:price-of-stability}
            \end{minipage}
\end{figure}

\begin{proof}
    We consider the network displayed in \Cref{fig:price-of-stability} with three vehicles, each with capacity $1$, and a single commodity with demand $2$, fixed departure time $\departureTimes\coloneqq\{ \sbTime{1}{00} \}$ and parameters $\beta=1$ and $\gamma^- = \gamma^+ = 0$.
    We assume that the outside option $\tau_{\outopt}$ is larger than the travel time of any path.

    The system-optimal flow sends one flow-unit along the \colourone{} path from~$s$ to~$t$ and one unit along the second \colourfour{} edge.
    The total social cost of this flow is $\pi(f)=1\cdot 4 + 1\cdot 5 = 9$.

    Here, the user equilibrium is unique:
    It sends one unit of flow along the \colourone-\colourfour $s$-$v$-$s$-$t$~path, and the remaining unit onto the \colourtwo path $p_r$.
    Thus, the price of stability equals $(4+\tau_{p_r})/9$.
    Increasing $\tau_{p_r}$ by delaying the arrival of the \colourtwo{} edge allows us to achieve an arbitrarily large price of stability.
\end{proof}

While the price of stability is unbounded in general networks, our computational study described in \Cref{sec:comp-study} suggests that this ratio is well-behaved in real-world networks.

\section{Computation of Equilibria}\label{sec:computation}

We continue by discussing the computation of user equilibria.
After describing an $\bigO(\abs{E}^2)$ time algorithm for single-commodity networks with fixed-departure times, we consider the general multi-commodity case, for which we outline a finite algorithm that computes a user equilibrium, if one exists, and that otherwise signals non-existence.
However, as we will show, determining whether a user equilibrium exists is NP-hard.
Thus, to compute multi-commodity equilibria in practice, we propose a heuristic based on insights gained by the characterization with the quasi-variational inequality. 

\subsection{An Efficient Algorithm for Single-Commodity Networks with Fixed Departure Times}\label{subsec:efficient-algo}

We begin with the description of an efficient algorithm for single-commodity networks with fixed departure times.
To reduce notational noise, we omit the index $i$ where applicable, e.g., we write $\pathsWithOutside$ instead of $\pathsWithOutside_i$.

\begin{definition}
    Let $p,q\in\pathsWithOutside$.
    We say that a driving edge $e\in p\cap q$ is a \emph{conflicting edge of $p$ and $q$} if its corresponding boarding edge $e_B$ lies either on $p$ or on $q$ (but not on both).

    Assume $p$ and $q$ have a conflicting edge, and let $e\in E$ be the first conflicting edge.
    We say \emph{$p$ has priority over $q$} if the boarding edge $e_B$ preceding $e$ lies on $q$ (and not on $p$).
    Let $\mathord{\prec} \subseteq \pathsWithOutside\times\pathsWithOutside$ denote this relation.
\end{definition}

\begin{proposition}\label{prop:minimal-path-no-bad-conflicting-edge}
    Let $p$ be a $\prec$-minimal path, i.e., there exists no $q\in\pathsWithOutside$ with $q\prec p$.
    Then for any $q\in\pathsWithOutside$, there is no conflicting edge $e$ of $p$ and $q$ for which the corresponding boarding edge $e_B$ lies on $p$.
\end{proposition}
\begin{proof}
    Let $q\in\pathsWithOutside$ be any path, and let $e$ be a conflicting edge of $p$ and $q$.
    As $q\not\prec p$, we know that $e$ is either not the first conflicting edge or the corresponding boarding edge does not lie on $p$.
    If $e$ is not the first conflicting edge, then let $e'$ denote the previous conflicting edge, and let $q'$ be the concatenation of the prefix of $p$ up to $e'$ and the suffix of $q$ starting from $e'$.
    Then $e$ is the first conflicting edge of $q'$ and $p$, and by $q'\not\prec p$ we know that the boarding edge of $e$ cannot lie on $p$.
\end{proof}

We first describe an efficient way to compute a $\prec$-minimal path ending in a given reachable node~$w$.
In fact, this can be done by a simple backward-search on the sub-graph of reachable nodes prioritizing dwelling edges over boarding edges, which is formalized in \cref{alg:compute-minimal-path}.

\begin{algorithm}[ht]
    \caption{Computes a $\prec$-minimal path} \label{alg:compute-minimal-path}
    \KwData{Time-expanded graph $G=(V,E)$, source station $s$, departure time interval $\departureTimes$, end node~$w$ of some path in $\pathsWithoutOutside$}
    \KwResult{A $\prec$-minimal path ending in $w$}
    $V'\gets $ nodes reachable in $G$ from any on-platform node $v$ of $s$ with $\theta(v)\in\departureTimes$\;
    $P \gets $ empty path\;
    $v \gets w$\;
    \While{$v\neq v^*$}{
     \lIf{$\exists e=u'v\in E\setminus E_B$ with $ u'\in V'$}{
      $u \gets u'$%
     }\lElse{
      \label{alg-line:choose-preceding} $u \gets $ any $u\in V'$ with $e=uv\in E$ %
     }
     $P \gets (uv) \circ P$\;
     $v \gets u$\;
    }
    \KwRet{$P$}
\end{algorithm}

\begin{lemma}\label{lemma:correctness-minimal-path-algo}
    For the end node~$w$ of any path in $\pathsWithOutside$, \Cref{alg:compute-minimal-path} computes a $\prec$-minimal path ending in $w$ in $\bigO(\abs E)$ time.
\end{lemma}
\begin{proof}
    To acknowledge that the algorithm is well posed, note that in line~\ref{alg-line:choose-preceding}, the node $v\neq v^*$ is always reachable from $v^*$ and thus, there is some preceding node $u$ that is also reachable from $v^*$.

    As the graph is acyclic, an edge can only be added once to the path $P$.
    Hence, the algorithm terminates after $\bigO(\abs E)$ time.

    It remains to show that the algorithm is correct.
    Let $p$ be the returned path, and let $q\in\pathsWithOutside$ be any path.
    Let $e=vw$ be the first conflicting edge of $p$ and $q$ -- if none exists, $q\not\prec p$ holds trivially.
    Clearly, the corresponding boarding edge of $e$ must lie on $q$: Otherwise, $q$ is evidence that there exists a non-boarding edge $e'=u'v$ with $u'\in V'$ and \Cref{alg:compute-minimal-path} would have chosen $e'$ for $p$ as well (contradicting that $e$ is a conflicting edge).
    Hence, also in this case, it holds $q\not\prec p$.
\end{proof}

In order to compute single-commodity equilibrium flows, we can now successively send flow along $\prec$-minimal and $\pi$-optimal paths.
In fact, we know that such a path exists, since for fixed departure times, the cost $\pi_p$ of a path $p$ only depends on its arrival time, i.e., there is some function $\pi': \R \to \R$ with $\pi'(\patharrival_p) = \pi_{p}$ for any path $p\in\pathsWithoutOutside$; see \cref{eq:pi-prime}.
Thus, we can simply choose an on-platform node~$w$ at $t$ that minimizes $\pi'(\theta(w))$ among the reachable on-platform nodes, and then find a $\prec$-minimal path ending in $w$ using \Cref{alg:compute-minimal-path}.

In every iteration, the flow on this path $p$ is increased until an edge becomes fully saturated.
Then, we reduce the capacity on the edges of $p$ by the added flow, remove zero-capacity edges, and repeat this procedure until the demand is met.
An explicit description is given in \Cref{alg:single-commodity} where we use the notation $\abs{f}\coloneqq \sum_{p\in\pathsWithOutside} f_p$.
\begin{algorithm}[ht]
    \caption{Computes user-equilibrium for single-commodity instances}\label{alg:single-commodity}
    \KwData{Time-expanded graph $G=(V, E)$, capacities $\nu\in\R_{>0}^{E_D}$, demand $\demand\in\R_{\geq 0}$, source and destination stations~$s, t$, departure time interval~$\departureTimes$, cost function $\pi':\R\to\R$\oxc and outside cost $\pi_{\outopt}\in\R$}
    \KwResult{A user equilibrium~$f$}%
    $f \gets \BFzero;$\quad$G'\gets G$;\quad $\nu'\gets \nu$\;
    \While{$\abs{f} < \demand$}{
      $W \gets $ on-platform nodes of $t$ reachable in $G'$ from on-platform nodes of $s$ during $\departureTimes$\;
    \If{$\inf_{w\in W} \pi'(\theta(w)) > \pi_{\outopt}$} {
        \KwRet{$f + (\demand - \abs{f})\cdot 1_{\outopt}$}\;
    }
     $p \gets$ $\prec$-minimal path w.r.t. $G'$ ending in $w$ for some $w\in\argmin_{w\in W}{\pi'(\theta(w))}$\;
     $\delta\gets\min(\{ \nu'_e \mid e\in p \}\cup\{\demand - \abs{f} \})$\;
     $f \gets f + \delta\cdot 1_p$\;
     $\nu' \gets \nu' - \delta\cdot 1_{p \cap E_D}$\;
     \For{$e\in p \cap E_D$ with $\nu'_e = 0$}{
      Remove $e$ and its incident edges from $G'$\;
     }
    }
    \KwRet{$f$}
\end{algorithm}
\begin{theorem}\label{thm:single_comp}
    For single-commodity networks, \Cref{alg:single-commodity} computes a user equilibrium in $\bigO(\abs E^2)$ time.
    The resulting user equilibrium uses at most $\abs E$ paths.
\end{theorem}
\begin{proof}
    To verify that the algorithm terminates, note that in each round after which the algorithm does not terminate, at least one driving edge is removed from the graph.
    Thus, there can be at most $\abs{E_D}$ many rounds.
    By \Cref{lemma:correctness-minimal-path-algo}, each round takes $\bigO(\abs{E})$ time.

    Let $f$ be the flow returned by the algorithm.
    Then $f$ is of the form \(
    f = \delta_1\cdot 1_{P_1} + \cdots + \delta_k\cdot 1_{P_k}
    \)
    where $k$ is the number of rounds, and $\delta_j > 0$, $P_j\in \pathsWithOutside$ are the values produced in round $j\in\{1,\dots,k\}$.
    By construction, we have $\pi_{P_j} \leq \pi_{P_l}$ and $P_j \neq P_l$ whenever $j< l$.
    We show by induction over $l$ that $f^l\coloneqq \sum_{j=1}^l \delta_j\cdot 1_{P_j}$ is a user equilibrium with demand~$\abs{f^l}$.
    The base case $l=0$ is trivial.
    Assume $f^l$ is a user equilibrium with demand~$\abs{f^l}$.
    Clearly, $f^{l+1}$ is capacity-feasible w.r.t. $\nu$.
    Let $p$ be a path with $f^{l+1}_p > 0$, and let $q\in\pathsWithOutside$ such that $\pi_q < \pi_p$.

    If $p=P_{j}$ for some $j<l+1$, then $f^l_p > 0$, and by induction hypothesis we have $q\notin \avPaths_p(f^l)$.
    As $\avPaths_p(f^{l+1})$ is a subset of $\avPaths_p(f^l)$, this implies $q\notin \avPaths_p(f^{l+1})$.
    Otherwise, we have $p=P_{l+1}$.
    Then $q$ is not present in $G'$ at the beginning of round $l+1$.
    Let $l^* \le l$ be the last round of the algorithm before which every edge of $q$ was still present in the graph $G'$.
    Clearly, $\tau_q \geq \tau_{P_{l^*}}$, as otherwise $P_{l^*}$ would not have been chosen in iteration~$l^*$.
    There is a driving edge $e$ on $q$ for which $f^{l^*}_e = \nu_e$.
    Let $e$ be the first such edge on $q$. Since flow is never removed from $e$, it holds that $f^{l+1}_e = \nu_e$.
    As $e$ is removed from $G'$ in round~$l^*$, it cannot be contained in $P_{l+1}$.
    Therefore, switching from $P_{l+1}$ to $q$ immediately creates a capacity violation. Hence, if the boarding edge preceding $e$ lies on $q$, we have $q\notin \avPaths_{P_{l+1}}(f^{l+1})$.
    Assume the boarding edge does not lie on $q$, i.e., $q$ uses the dwelling edge $e'$ before traversing $e$.
    As $q$ was contained in $G'$ at rounds $l^j$, $j\in\{1,\dots,l^*\}$, we know by \Cref{prop:minimal-path-no-bad-conflicting-edge} that $e$ is not a conflicting edge of $q$ and $P^{j}$, and thus, if $P_j$ uses $e$, it must also use the same dwelling edge $e'$ and the previous driving edge~$e''$ of the vehicle. Hence, $f_{e''}^{l^*} = f_e^{l^*} = \nu_e = \nu_{e''}$.
    This contradicts the minimality of the position of $e$ in $q$.
\end{proof}

For general (aperiodic) schedules our algorithm is strongly polynomial in the input.
For compactly describable periodic schedules it is only pseudo-polynomial as it depends on the size of the time-expanded network.
The blow-up of the network depends on the ratio of time horizon and period length.

\begin{remark}
    The more general single-destination scenario, where we allow multiple commodities sharing a common destination station, can be reduced to the single-commodity case by introducing a so-called super-source node: For each original commodity we add a connection to the node of the commodity's station and start time, which is capacitated by the commodity's demand.
\end{remark}

\subsection{The General Multi-Commodity Case with Departure Time Choice}

The approach of the previous section fails for the general multi-commodity case as the set of paths $\bigcup_{i} \pathsWithOutside_i$ may not necessarily have a $\prec$-minimal element if there are commodities that do not share the same destination station (see~\Cref{fig:cyclic-behavior} for an example).
In the following, we describe a finite-time algorithm that works for multi-commodity networks and departure time choice.
As we have seen, user equilibria might not exist in the case of departure time choice, in which case the algorithm will terminate with a certificate of non-existence.
As the algorithm has an exponential runtime, we propose a heuristic for computing multi-commodity equilibria in practice.
Furthermore, we show that determining whether a user equilibrium exists is NP-hard.

\subsubsection{A finite-time algorithm} \label{subsec:finite-time algorithm}

In the following, we describe a finite-time algorithm for computing exact multi-commodity user equilibria.
Assuming an equilibrium $f$ exists, the idea is to guess the subset~$E_S$ of driving edges that are saturated, i.e., $E_S = \{ e \in E_D \mid f_e = \nu_e \}$.
If a user equilibrium saturating these edges exists, it can be found by solving a set of linear constraints.

\newcommand*{\flowset}{\mathcal F}
More specifically, we define the following set of feasible flows
\begin{equation}
    \flowset(E_S) \coloneqq \left\{
        f\in\demFeasFlows
        \middle|
        \begin{array}{c l}
            f_e = \nu_e,                  & \text{ for $e\in E_S$,}                               \\
            f_e \leq \nu_e,               & \text{ for $e \in E_D \setminus E_S$,}                \\
            f_{i,p} = 0,                  & \text{ for $i\in I$, $p\in \pathsWithOutside_i(E_S)$}
        \end{array}
    \right\},
\end{equation}
where $\pathsWithOutside_i(E_S)$ is the set of paths $p\in\pathsWithOutside_i$ for which there exists a better alternative path $q\in\pathsWithOutside_i$ with respect to the saturated edge set $E_S$.
More specifically, \[
    p\in \pathsWithOutside_i(E_S) \;:\!\Longleftrightarrow\; \exists q\in \pathsWithOutside_i : \pi_{i,q} < \pi_{i,p} \,\land\, \forall e\in E_B\cap q: e^+ \notin E_S \lor e^+ \in p.
\]

\begin{lemma}\label{lem:finite-algo}
    The set of user equilibria coincides with the union of $\flowset(E_S)$ over all $E_S\subseteq E_D$.
\end{lemma}
\begin{proof}
    Note that for any feasible flow $f$ and path $p\in \pathsWithOutside_i$, there exists some $q\in\avPaths_{i,p}(f)$ with $\pi_{i,q} < \pi_{i,p}$ if and only if $p\in\pathsWithOutside_i(E_S^f)$ with $E_S^f \coloneqq \{ e\in E_D \mid f_e = \nu_e \}$.
    Thus, $f$ is a user equilibrium if and only if $f$ is feasible and $f_{i,p}=0$ holds for all $p\in \pathsWithOutside_i(E_S^f)$, $i \in I$.

    Now, if $f$ is a user equilibrium, it is clearly contained in $\flowset(E_S^f)$.
    Conversely, for a given set $E_S$ and a flow $f\in \flowset(E_S)$, we know that $E_S \subseteq E_S^f$ and therefore $\pathsWithOutside_i(E_S) \supseteq \pathsWithOutside_i(E_S^f)$.
    This implies that $f_{i,p}=0$ holds for all $p\in \pathsWithOutside_i(E_S^f)$, and thus $f$ is a user equilibrium.
\end{proof}

We can check the feasibility of $\flowset(E_S)$ for every subset $E_S$ of $E_D$ in finite time.
If a user equilibrium exists, we will find it; otherwise, we can terminate with the certificate that no user equilibrium exists.

\begin{corollary}
    The procedure described above checks in finite time, whether a user equilibrium exists or not, and returns one, if it exists.
\end{corollary}

\subsubsection{NP-Hardness}\label{subsec:NP-hardness}

\newcommand{\ttrue}{\textsc{true}}
\newcommand{\tfalse}{\textsc{false}}

\newcommand*{\myproblem}[4]{%
        \begin{problem}{#1}\label{#2}
            \begin{tabularx}{\textwidth}{rX}
                \emph{Input:} & #3 \\
                \emph{Question:} & #4
            \end{tabularx}
        \end{problem}%
}

In this section, we address the computational complexity of computing multi-commodity user equilibria.
While \Cref{thm:existence-metro} guarantees the existence of a user equilibrium for the case of fixed departure times, we will show in this section that it is NP-hard to decide whether a user equilibrium exists when allowing departure time choice.

\myproblem{UE-DTC}{prob:ue-dtc}{A time-expanded graph $G$, a finite set of commodities $I$.}{Is there a user equilibrium?}

Even when restricting to fixed departure time instances, we show that it is NP-hard to decide whether a user equilibrium exists that fulfils certain properties.
In particular, we will show that the following problems are NP-hard as well:

\myproblem{UE-IN}{prob:inside-ue}{A time-expanded graph $G$, a finite set of commodities $I$ with fixed departure times.}{Is there a user equilibrium in which no particle uses its outside option?}

\myproblem{UE-OPT}{prob:opt-ue}{A time-expanded graph $G$, a finite set of commodities $I$ with fixed departure times.}{Is there a user equilibrium that is also a system optimum?}

\myproblem{UE-SCT}{prob:sct}{A time-expanded graph $G$, a finite set of commodities $I$ with fixed departure times, threshold $C$.}{Is there a user equilibrium with social cost at most $C$?}

We show NP-hardness by polynomially reducing 3-SAT to these problems.
A problem instance of 3-SAT consists of a set of $n$ boolean variables $x_1,\dots, x_n$ and a set of $m$ clauses $C_1,\dots,C_m$ where each clause $C_j$ is a disjunction of up to three literals (a variable $x_i$ or its negation $\overline{x_i}$).
The associated question is whether there exists a variable assignment such that all clauses are satisfied.

Given such a 3-SAT instance, we now construct a time-expanded graph $G$ and a set of commodities~$I$ (both polynomial in the input size of the 3-SAT instance).
Without loss of generality, we may assume that no clause contains a variable and its negation at the same time: These clauses are always fulfilled and may be discarded.

For each variable $x_i$ we define a commodity with a demand volume of $1$ and with new origin and destination stations $s_{x_i}$ and $t_{x_i}$, respectively.
Furthermore, for each variable, we introduce two vehicles -- a \colourthree and a \colourtwo one -- that both start at $s_{x_i}$ and end at $t_{x_i}$ at the same times, and both have a capacity of~$1$.
Here, the \colourthree{} and the \colourtwo{} vehicles represent the states in which $x_i$ is set to \ttrue{} and \tfalse, respectively.
The journey of the two vehicles between their departure at $s_{x_i}$ and their arrival at $t_{x_i}$ will be defined later.
The commodity has the fixed departure time matching the departure time of the vehicles.
The path costs are given by their travel time, i.e., $\pi_{i,p}\coloneqq \tau_p$ for all $p\in\pathsWithoutOutside$, and we set the cost of the outside option for the commodity to some number larger than the maximum travel time of any path in $\pathsWithoutOutside$.
As the capacity of both vehicles is $1$, this commodity will never use its outside option in any user equilibrium.
\Cref{fig:variables-partial-network} illustrates the described partial network for $n=4$ variables.

\newcommand{\variablegadgets}{
 \begin{tikzpicture}[y=0.7cm, x=1.1cm, thick]
  \newcommand{\height}{1.5}
  \newcommand{\depxoffset}{0.3}

  \lifelineCompactC{3}{$t_{x_1}$}
  \lifelineCompactC{5}{$t_{x_2}$}
  \lifelineCompactC{7}{$t_{x_3}$}
  \lifelineCompactC{9}{$t_{x_4}$}

  \begin{scope}[yshift=3cm]
      \lifelineCompactC{3}{$s_{x_1}$}
      \lifelineCompactC{5}{$s_{x_2}$}
      \lifelineCompactC{7}{$s_{x_3}$}
      \lifelineCompactC{9}{$s_{x_4}$}
  \end{scope}

  \foreach \x in {3,5,7,9} {
      \draw[draw=color2] (\x, 4.25) to[out=225, in=90] (\x - 0.5, 3.25);
      \draw[draw=color2, dashed] (\x - 0.5, 3.25) to (\x - 0.5, 1.25);
      \draw[->, draw=color2] (\x - 0.5, 1.25) to[out=-90, in=135] (\x, 0.25);

      \draw[draw=color3] (\x, 4.25) to[out=-45, in=90] (\x + 0.5, 3.25);
      \draw[draw=color3, dashed] (\x + 0.5, 3.25) to (\x + 0.5, 1.25);
      \draw[->, draw=color3] (\x + 0.5, 1.25) to[out=-90, in=45] (\x, 0.25);
  }
 \end{tikzpicture}
}

\begin{figure}[hb]
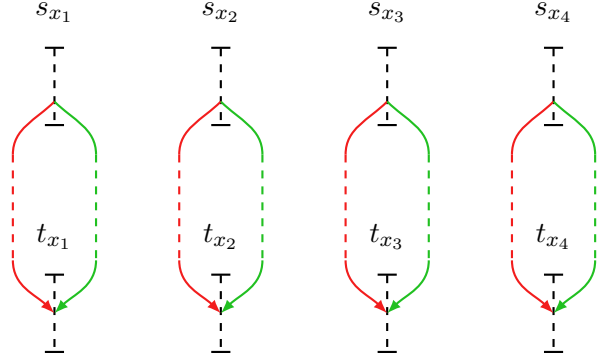

            \centering
        \variablegadgets
        \caption{Origins and destinations for the commodities induced by the variables.}\label{fig:variables-partial-network}
    \end{figure}

We carefully design our network such that the following property will be satisfied:
    \begin{enumerate}[label=(\Alph*)]
        \item\label{inv:no-mixing} For every variable $x_i$, every path from $s_{x_i}$ to $t_{x_i}$ (excluding the outside option) either exclusively uses the corresponding \colourthree{} vehicle or exclusively uses the corresponding \colourtwo{} vehicle.
    \end{enumerate}

This invariant helps us to ensure that no particle of variable $x_i$ mixes between the \ttrue{} and \tfalse{} states nor (directly) influences any other vehicles on their way to their destination.

We now construct a gadget with a corresponding commodity for each clause $C_j$ of the 3-SAT instance that will use the green and red vehicles of the variables in such a way that the clause is fulfilled if and only if in the corresponding equilibrium no particle of the clause's commodity uses its outside option.
For this, we use an adaptation of the network in \Cref{fig:price-of-stability} where we interpret the last, red vehicle as the outside option of the commodity.
Recall that, in this setting the demand of the single commodity is $2$ units and the unique user equilibrium is given by sending one unit of flow along the \colourone{}-\colourfour{} $s$-$v$-$s$-$t$ path and the remaining flow along its outside option.
However, if the journey from $v$ back to $s$ was blocked (e.g. if the \colourfour{} vehicle started prior to $v$ and if it was already fully occupied when arriving at $v$), then the particles of the considered commodity would split between the \colourone{} $s$-$v$-$t$-path and the \colourfour{} $s$-$t$-path with $1$ unit each; thus no particle would have to use the outside option.
The idea is now to use the green and red vehicles of the variables to block the journey from $v$ back to $s$ depending on their occurrence in the clause.

More specifically, for a given clause~$C_j$ we define a commodity with a demand volume of $2$ and introduce a set of new stations $s_{C_j}$, $v^0_{C_j}$, \dots, $v^n_{C_j}$, $t_{C_j}$ (arranged from left to right), where $s_{C_j}$ and $t_{C_j}$ serve as the origin and destination stations of the new commodity, respectively.
The gadget is placed (temporally) between the departure at $s_{x_i}$ and arrival at $t_{x_i}$ of the variables' vehicles and thus may modify the journey of these vehicles within the gadget's relevant time period.
We introduce two vehicles (of capacity $1$): a \colourone vehicle leading from $s_{C_j}$ via $v^n_{C_j}$ to $t_{C_j}$ and a \colourfour vehicle leading from $v^0_{C_j}$ via $s_{C_j}$ to $t_{C_j}$ such that the \colourfour vehicle arrives earlier at $t_{C_j}$ than the \colourone vehicle.

Next, we construct a path that connects particles arriving at $v^n_{C_j}$ on the \colourone vehicle to the \colourfour vehicle departing from $v^0_{C_j}$.
This path uses $n$ different vehicles, each covering a step from $v^i_{C_j}$ to $v^{i-1}_{C_j}$:
If the variable $x_i$ appears as a positive literal in the clause, we use the \colourthree{} vehicle of $x_i$ for this step; if $x_i$ appears as a negative literal, we use $x_i$'s \colourtwo{} vehicle; otherwise we add a new vehicle of capacity~$1$.
Finally, the path costs of this commodity coincide with the paths' travel times, i.e., $\pi_{C_j,p} = \tau_{C_j,p}$, and the outside option cost is some time larger than the maximum of these travel times.
The set of feasible departure times $\Theta_i$ will be specified later.
\Cref{fig:clause-gadget} illustrates the described gadget for a sample clause $C_j= (x_1\lor \overline{x_2}\lor\overline{x_4})$ with $n=4$.

\newcommand{\clausegadget}{
 \begin{tikzpicture}[y=0.7cm, x=1.1cm, thick]
  \newcommand{\height}{8}
  \newcommand{\depxoffset}{0.3}

  \lifelineCompactC{0}{$s_{C_j}$}
  \lifelineCompactC{2}{$v^0_{C_j}$}
  \lifelineCompactC{4}{$v^1_{C_j}$}
  \lifelineCompactC{6}{$v^2_{C_j}$}
  \lifelineCompactC{8}{$v^3_{C_j}$}
  \lifelineCompactC{10}{$v^4_{C_j}$}
  \lifelineCompactC{12}{$t_{C_j}$}

  \draw[->, draw=color1] (0, \height - 1) -- (10, \height - 1.67);
  \draw[->, draw=color1] (10, \height - 1.67) -- (12, \height - 8);

  \draw[->, draw=color3, dash pattern=on 3pt off 3pt on 3pt off 3pt on 3pt off 3pt on 10000pt] (3.15, \height + 0.5) to [out=-90,in=135] (4, \height - 5.33);
  \draw[->, draw=color3] (4, \height - 5.33) -- (2, \height - 5.67);
  \draw[draw=color3] (2, \height - 5.67) to [out=-45,in=90] (3.15, \height - 8.5);
  \draw[->, draw=color3, dashed] (3.15, \height - 8.5) -- (3.15, \height - 8.5  - 1);

  \draw[draw=color2, dash pattern=on 3pt off 3pt on 3pt off 3pt on 3pt off 3pt on 10000pt] (2.85, \height + 0.5) -- (2.85, \height - 8.5);
  \draw[->, draw=color2, dashed] (2.85, \height - 8.5) -- (2.85, \height - 8.5  - 1);

  \draw[draw=color3, dash pattern=on 3pt off 3pt on 3pt off 3pt on 3pt off 3pt on 10000pt] (5.15, \height + 0.5) -- (5.15, \height - 8.5);
  \draw[->, draw=color3, dashed] (5.15, \height - 8.5) -- (5.15, \height - 8.5  - 1);

  \draw[->, draw=color2, dash pattern=on 3pt off 3pt on 3pt off 3pt on 3pt off 3pt on 10000pt] (4.85, \height + 0.5) to [out=-90,in=135] (6, \height - 4.33);
  \draw[->, draw=color2] (6, \height - 4.33) -- (4, \height - 4.67);
  \draw[draw=color2] (4, \height - 4.67) to [out=-45,in=90] (4.85, \height - 8.5);
  \draw[->, draw=color2, dashed] (4.85, \height - 8.5) -- (4.85, \height - 8.5  - 1);

  \draw[draw=color3, dash pattern=on 3pt off 3pt on 3pt off 3pt on 3pt off 3pt on 10000pt] (7.15, \height + 0.5) -- (7.15, \height - 8.5);
  \draw[->, draw=color3, dashed] (7.15, \height - 8.5) -- (7.15, \height - 8.5  - 1);

  \draw[draw=color2, dash pattern=on 3pt off 3pt on 3pt off 3pt on 3pt off 3pt on 10000pt] (6.85, \height + 0.5) -- (6.85, \height - 8.5);
  \draw[->, draw=color2, dashed] (6.85, \height - 8.5) -- (6.85, \height - 8.5  - 1);

  \draw[->, draw=black] (8, \height - 3.33) -- (6, \height - 3.67);

  \draw[draw=color3, dash pattern=on 3pt off 3pt on 3pt off 3pt on 3pt off 3pt on 10000pt] (9.15, \height + 0.5) -- (9.15, \height - 8.5);
  \draw[->, draw=color3, dashed] (9.15, \height - 8.5) -- (9.15, \height - 8.5  - 1);

  \draw[->, draw=color2, dash pattern=on 3pt off 3pt on 3pt off 3pt on 3pt off 3pt on 10000pt] (8.85, \height + 0.5) to [out=-90,in=135] (10, \height - 2.33);
  \draw[->, draw=color2] (10, \height - 2.33) -- (8, \height - 2.67);
  \draw[draw=color2] (8, \height - 2.67) to [out=-45,in=90] (8.85, \height - 8.5);
  \draw[->, draw=color2, dashed] (8.85, \height - 8.5) -- (8.85, \height - 8.5  - 1);

  \draw[->, draw=color4] (2, \height - 6.33) -- (0, \height - 6.67);
  \draw[->, draw=color4] (0, \height - 6.67) -- (12, \height - 7.33);

 \end{tikzpicture}
}

\begin{figure}[ht]
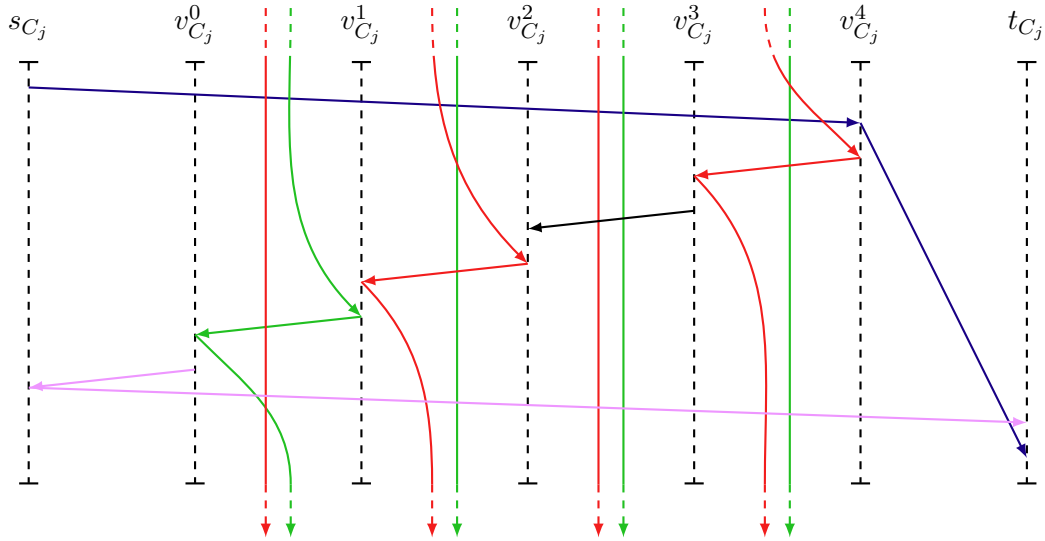

            \centering
        \clausegadget
        \caption{Gadget for the clause $C_j = (x_1 \lor \overline{x_2} \lor \overline{x_4})$.}\label{fig:clause-gadget}
    \end{figure}

The final network is then constructed by simply temporally stacking first all variables' origins, then one gadget for each clause, and finally the variables' destinations.

\begin{lemma}\label{prop:final-network-separable}
    The final network fulfils property~\ref{inv:no-mixing}.
\end{lemma}
\begin{proof}
    A path from $s_{x_i}$ to $t_{x_i}$ either uses the \colourthree{} vehicle or the \colourtwo{} vehicle when departing in $s_{x_i}$.
    In the gadgets of the clauses, it is never possible to alight from the used vehicle and board any other vehicle while still being able to reach $t_{x_i}$.
\end{proof}

We have an analogous property for the clause commodities.

    \begin{enumerate}[label=(\Alph*), resume]
        \item\label{inv:no-escaping-clause-flows} For every clause~$C_j$, every path from $s_{C_j}$ to $t_{C_j}$ (excluding the outside option) exclusively uses the edges in its associated gadget.
    \end{enumerate}

\begin{theorem}\label{thm:ue-dtc-np-hard}
    \Cref{prob:ue-dtc} is NP-hard.
\end{theorem}
\begin{proof}
    Consider the final network as described above and assume that the commodities which correspond to clauses have free departure time choice, i.e., $\departureTimes_i=\R$.
    For the commodity corresponding to clause $C_j$, there are three relevant paths:
    The \colourone{} path~$p_{j,1}$ from $s_{C_j}$ via $v^n_{C_j}$ to $t_{C_j}$ with the highest travel time,
    the zig-zag-path $p_{j,2}$ from $s_{C_j}$ via $v_{C_j}^n$ and $s_{C_j}$ to $t_{C_j}$ with the second-highest travel time,
    and the \colourfour{} path $p_{j,3}$ that directly connects $s_{C_j}$ to $t_{C_j}$ which comes with a later departure time and the least travel time.\par{}
    Note that in a user equilibrium, the zig-zag-path $p_{j,2}$ is never used:
    Otherwise, $p_{j,3}$ would be a better available path.

    \begin{claim}
        There exists a user equilibrium if and only if the given $3$-SAT instance is satisfiable.
    \end{claim}
    \begin{proofClaim}
        Let $f$ be a user equilibrium.
        We assign the variable $x_i$ the value \ttrue{} if exactly $1$ unit of flow boards its corresponding \colourthree{} vehicle at $s_{x_i}$, otherwise \tfalse{}.
        We now consider the flow in the gadget of a clause $C_j$.
        Note that the particles of the clause's commodity do not use their outside option, as otherwise at least one of the two paths $p_{j,1}$ or $p_{j,3}$ would be an available alternative with a lower cost.
        As observed above, the zig-zag-path $p_{j,2}$ does not carry any flow, and hence, the flow of the clause's commodity must split between the paths $p_{j,1}$ and $p_{j,3}$ each carrying exactly $1$ unit of flow.
        This also means that $p_{j,2}$ is not an available alternative to the slower path $p_{j,1}$.
        This is only the case if at least one driving edge on the path from $v_{C_j}^n$ to $v_{C_j}^1$ is fully occupied by flow from a different commodity.
        By Properties~\ref{inv:no-mixing} and \ref{inv:no-escaping-clause-flows}, the only flow using these edges comes from the commodities of the variables included in the clause.
        Therefore, one of the \colourthree{} and \colourtwo{} vehicles on the path, corresponding to some variable $x_i$, carries exactly $1$ unit of flow of $x_i$'s commodity.
        By Property~\ref{inv:no-mixing}, this means that this $1$ unit of flow must have boarded the same vehicle at $s_{x_i}$.
        Hence, the corresponding literal in the clause~$C_j$ (and thus the clause itself) is satisfied.

        Conversely, assume there is a variable assignment satisfying the $3$-SAT instance, and consider the flow $f$ defined as follows:
        For a variable $x_i$ we assign all flow of its commodity to its \colourthree{} path if $x_i$ is assigned the value \ttrue{}, and all flow to the \colourtwo{} path, otherwise.
        For a clause $C_j$, we send one unit of flow along $p_{j,1}$ and one unit of flow along $p_{j,2}$.
        This flow is clearly feasible.
        To verify the equilibrium condition, we only need to check particles of the clause commodities using the \colourone{} path $p_{j,1}$.
        Since there is a satisfied literal in the clause, the driving edge of the corresponding vehicle is fully occupied by the flow of the variable's commodity.
        Hence, this edge makes the zig-zag-path $p_{j,2}$ unavailable as an alternative to $p_{j,1}$.
        Thus, $f$ is in fact a user equilibrium.
    \end{proofClaim}

    This claim reduces $3$-SAT to the problem of deciding whether a user equilibrium (with departure time choice) exists, and 
    since the constructed network has polynomial size, the latter problem is NP-hard.
\end{proof}

\begin{theorem}
    \Cref{prob:inside-ue} is NP-hard.
\end{theorem}
\begin{proof}
    We use the same network as above, but we restrict the departure times of the commodities corresponding to the clauses to $\Theta_j = \{0\}$.
    It suffices to show that there exists a user equilibrium in which no particle uses its outside option if and only if the given 3-SAT instance is satisfiable.

    Let $f$ be a user equilibrium in which no particle uses its outside option.
    Again, we assign to the variable $x_i$ the value \ttrue{} if exactly $1$ unit of flow boards its corresponding \colourthree{} vehicle at $s_{x_i}$, otherwise we assign the value \tfalse{}.
    Consider a clause $C_j$ and the flow in its corresponding gadget.
    As no particle uses its outside option, the flow of the clause's commodity must split between the gadget's \colourone{} and \colourfour{} vehicles each carrying exactly $1$ unit of flow.
    This means, the zig-zag-path from $s_{C_j}$ via $v_{C_j}^n$ and $s_{C_j}$ to $t_{C_j}$ is not an available alternative to the path using only the \colourone{} vehicle.
    Following the same arguments of the proof of \Cref{thm:ue-dtc-np-hard}, this implies that the clause $C_j$ must be fulfilled, and the $3$-SAT instance is satisfiable.

    Conversely, assume that the 3-SAT instance is satisfiable, and let $\sigma$ be an assignment of the variables that satisfies all clauses.
    Just as above, we construct the flow $f$ by sending, for each variable, all particles along the \colourthree{} vehicle if $\sigma$ sets the variable to \ttrue{} and along the \colourtwo{} vehicle otherwise;
    for each clause, we send $1$ unit of flow along the \colourone, $s_{C_j}$-$v^n_{C_j}$-$t_{C_j}$-path and $1$ unit of flow along the faster \colourfour, $s_{C_j}$-$t_{C_j}$-path.

    This flow does not use any outside option, and it remains to verify the equilibrium condition.
    For this, we only need to consider particles of the clause commodities using the \colourone{} path.
    There is a literal in the clause that is satisfied.
    As the clause $C_j$ is fulfilled, there is at least one driving edge in the zig-zag-path that is fully occupied.
    Therefore, this flow is a user equilibrium.
\end{proof}

Note that in the considered final network in the fixed departure time setting, a system-optimal flow always splits the flow of a clause's commodity between the \colourone{} and \colourfour{} vehicles with one unit each.
In particular, in this network, the set of user equilibria that do not use any outside option coincides with the set of user equilibria that are also system optima.
This proves the following corollary:

\begin{corollary}
    \Cref{prob:opt-ue} is NP-hard.
\end{corollary}

Since the system-optimal social cost can be computed in polynomial time by solving a linear program, this implies that the decision problem of whether a user equilibrium exists whose social cost is smaller or equal to a given value is also NP-hard:

\begin{corollary}
    \Cref{prob:sct} is NP-hard.
\end{corollary}

\subsubsection{Heuristic for computing multi-commodity equilibria}

The previous section shows that deciding whether a user equilibrium exists in a general multi-commodity setting is NP-hard.
Also, for practical applications, the algorithm described in \Cref{subsec:finite-time algorithm} is not tractable due to its exponential running time.
For this reason, we propose a heuristic for computing multi-commodity user equilibria in the following.
In \Cref{sec:comp-study}, we will evaluate the performance of this heuristic on large-scale real-world networks.

The heuristic works as follows: 
Start with some initial feasible flow $f\in\feasFlows$, e.g., by sending all flow along their outside option.
Then, iteratively, find a direction $d\in\R^{\pathsWithOutside}$ and change the flow along this direction while preserving feasibility, until an equilibrium is found.
More specifically, we replace $f$ with $f'=f+\lambda\cdot d$ where~$\lambda$ is the maximal value such that $f'$ is feasible.
\begin{definition}
    Let $f$ be a feasible flow.
    A direction $d\in\R^{\pathsWithOutside}$ is called
    \begin{itemize}
        \item \emph{balanced} if $\sum_{p\in\pathsWithOutside_i} d_{i,p} = 0$ for $i\in I$, and
        \item \emph{feasible for $f$} if the flow $f+\lambda\cdot d$ is feasible for small enough $\lambda > 0$.
    \end{itemize}
\end{definition}

Clearly, the choice of the direction is essential for this heuristic to approach an equilibrium.
The characterization in \Cref{thm:QVI} indicates using a direction $d$ such that $f+\lambda\cdot d$ is a deviation violating the quasi-variational inequality.
This means that we should use some direction $d\coloneqq 1_{i,q} - 1_{i,p}$ for some paths $p,q\in\pathsWithOutside_i$ with $f_{i,p} > 0$ for which~$q$ is a better available alternative, i.e., $q\in \avPaths_{i,p}(f)$ and $\pi_q < \pi_p$.
However, not all such directions are feasible; even worse, sometimes no feasible direction is of this form.
Therefore, our approach is to start with such a direction~$d$ and, if necessary, transform it to make it feasible.
See~\Cref{alg:heuristic-multi-commodity} for a sketch of the method described so far.

\begin{algorithm}[ht]
    \caption{Heuristic for computing multi-commodity equilibria}
    \label{alg:heuristic-multi-commodity}
    \KwData{Time-expanded graph $G=(V,E)$, finite set of commodities $I$}
    \KwResult{A user equilibrium}
    Initialize $f$ by sending all flow along outside options\;
    \While{$\exists i\in I, p\in\pathsWithOutside_i,  q\in \avPaths_{i,p}(f)$ with $\pi_{i,q} < \pi_{i,p}$}{
     $d\gets 1_{i,q} - 1_{i,p}$\;
     (Potentially) transform $d$ to a feasible direction of $f$\;
     $f\gets f+\lambda \cdot d$ with $\lambda>0$ maximal such that $f+\lambda \cdot d$ is feasible\;
    }
    \KwRet{$f$}
\end{algorithm}

If this heuristic terminates, it provides an equilibrium, but termination is not always guaranteed, as we will see later.
We first describe how we achieve feasibility of the direction.
For this, a key observation is stated in the following proposition:

\begin{restatable}{proposition}{PropFeasibleDirectionBoardingEdge}\label{prop:feasible-direction-boarding-edge}
    Let $f$ be a feasible flow and $d$ a balanced direction that fulfils $f_{i,p} > 0$ whenever $d_{i,p} < 0$.
    Then, $d$ is a feasible direction for $f$ if and only if there exists no boarding edge $e$ such that $f_{e^+} = \nu_{e^+}$, $d_{e^+} > 0$ and ($f_e > 0$ or $d_{e} > 0$) hold.
\end{restatable}
\newcommand{\proofPropFeasibleDirectionBoardingEdge}[1][Proof]{
\begin{proof}[#1]
    Assume $d$ is a feasible direction for $f$, and let $e$ be any driving edge with $f_{e} = \nu_{e}$.
    Then, by the feasibility of $f+\lambda\cdot d$ for small enough $\lambda$, we must have $d_{e}\leq 0$.

    Assume now that $d$ is not a feasible direction for $f$.
    Then, there is some $i$ and a path $p$ such that $f_{i,p} + \lambda \cdot d_{i,p} < 0$ for all $\lambda > 0$, which is impossible due to our assumption on $d$, or there is a driving edge $e^+$ following some boarding edge $e$ such that $f_{e^+} + \lambda \cdot d_{e^+} > \nu_{e^+}$ for all $\lambda > 0$.
    Let $e^+$ be the first such driving edge in its corresponding vehicle.
    As $f$ is feasible, we must have $f_{e^+} = \nu_{e^+}$ and $d_{e^+} > 0$.
    If both $f_e = 0$ and $d_e\leq 0$ held, the driving edge $\tilde e^+$ of the previous stop of the same vehicle would fulfil $f_{\tilde e^+} = f_{e^+} = \nu_{e^+} = \nu_{\tilde e^+}$ and $d_{\tilde e^+}\geq d_{e^+} > 0$, a contradiction to the minimality of $e^+$.
\end{proof}
}
\proofPropFeasibleDirectionBoardingEdge

If $d=1_{i,q} - 1_{i,p}$ is an infeasible direction, we apply the following transformation:
As long as $d$ is infeasible, there exists a boarding edge such that $f_{e^+} = \nu_{e^+}$, $d_{e^+} > 0$\oxc and $f_{e} > 0 \,\lor\, d_e > 0$, and we repeat the following procedure:
Let $(i,p')$ be such that $p' \in \pathsWithOutside_i$ is a path containing $e$ either with positive flow $f_{i,p'} > 0$ or whose entry in the direction vector is positive, i.e., $d_{i,p'} > 0$.
We decrease $d_{i,p'}$ by $\delta\coloneqq d_e^+$ if $f_{i,p'} >0$, or by $\delta\coloneqq \min(d_e^+, d_{i,p'})$ otherwise.
Next, we determine a best path $q'\in\pathsWithOutside_i$ that does not use full driving edges, i.e., driving edges $\tilde e$ with $f_{\tilde e} = \nu_{\tilde e}$ and $d_{\tilde e} \geq 0$.
We increase $d_{i,q'}$ by $\min ( \{ \delta \} \cup \{ -d_e \mid e\in q', f_e = \nu_e \} )$, and afterwards decrease $\delta$ by the same amount.
We repeat this until $\delta$ is zero.
\Cref{alg:remove-overfilled-edges} describes this transformation of the direction in detail.

\begin{algorithm}[ht]
    \caption{Establishing feasible directions}
    \label{alg:remove-overfilled-edges}
    \KwData{Time-expanded graph with outside options, feasible flow $f$, balanced direction~$d\in\Z^\pathsWithOutside$ s.t. $\forall (i,p): d_{i,p} < 0\implies f_{i,p} > 0$}
    \KwResult{A feasible direction}
    \While{$\exists e\in E_B$ with $  f_{e^+} = \nu_{e^+} \,\land\, d_{e^+} > 0 \,\land\, (f_e > 0 \,\lor\, d_e > 0)$}{
     $(i,p)\gets$ any commodity $i$ and path $p$ containing $e$ with $f_{i,p} > 0$ or $d_{i,p} > 0$\;
    \label{alg-line:set-delta} $\delta\gets \begin{cases}
            d_{e^+},              & \text{if $f_{i,p} > 0$}, \\
            \min(d_{e^+}, d_{i,p}), & \text{otherwise.}
        \end{cases}$\;
     Decrease $d_{i,p}$ by $\delta$\;
     \While{$\delta > 0$}{
      $q\gets$ best alternative to $p$ not containing any $e'\in E_D$ with $f_{e'}=\nu_{e'} \,\land\, d_{e'} \geq 0$\;
      $\delta'\gets \min ( \{ \delta \} \cup \{ -d_{e'} \mid e'\in q, f_{e'} = \nu_{e'} \} )$\;
      Increase $d_{i,q}$ by $\delta'$\;
      Decrease $\delta$ by $\delta'$\;
     }
    }
    \KwRet{$d$}
\end{algorithm}

\vspace{-1ex}
\begin{restatable}{proposition}{PropRemoveOverfilledEdgesCorrect}\label{prop:remove-overfilled-edges-correct}
    \Cref{alg:remove-overfilled-edges} transforms any direction $d$ that fulfils $f_{i,p} > 0$ whenever $d_{i,p} < 0$ to a feasible direction.
\end{restatable}
\newcommand{\proofPropRemoveOverfilledEdgesCorrect}[1][Proof]{
\begin{proof}[#1]
    \Cref{prop:feasible-direction-boarding-edge} implies correctness.
    For termination, note that the inner loop always terminates as a path $q$ is selected at most once.
    For edges $e'\in E_D$ with $f_{e'} = \nu_{e'}$, while $d_{e'}$ is positive, it is monotonically decreasing with progression of the algorithm, and once it is non-positive it will never become positive again.
    In the main loop, every pair $(e, p)$ is considered at most once:
    Clearly, this is the case if $\delta$ is set to $d_{e^+}$ in line~\ref{alg-line:set-delta} as then the edge $e$ will never be considered again because the direction $d_{e^+}$ will never become positive again.
    Otherwise, $\delta$ is set to $d_{i,p}$ and, subsequently, $d_{i,p}$ is set to $0$.
    After that, $d_{i,p}$ will not become positive unless $d_{e^+}$ is non-positive.
\end{proof}
}
\proofPropRemoveOverfilledEdgesCorrect
In the remainder of this section, we analyse the heuristic in more detail:
We first discuss some undesired behaviour of the heuristic and then present a technique to reduce the input complexity.

\subsubsection*{Undesirable behaviour}

In some situations, the heuristic might apply changes along directions $d_1,\dots, d_k$ in a cyclic behaviour.
We distinguish between \emph{terminating} cycles, for which the heuristic breaks out of the cyclic behaviour after some finite but potentially large number of iterations, and \emph{non-terminating} cycles.
In practice, most terminating cycles can be detected and prohibited by changing the flow along the common direction $\sum_{i=1}^k d_i$, thereby skipping a potentially large number of iterations.
Non-terminating cycles, however, constitute a more serious problem.
We can detect these cycles, as their common direction $\sum_{i=1}^k d_i$ vanishes.
Randomizing the path selection in the main loop of the heuristic might help in breaking the cycle.
However, in some cases, even this is impossible, requiring us to restart the heuristic with a randomized path selection rule.

\begin{example}[A terminating cycle]

\newcommand{\smallsteps}{
 \begin{tikzpicture}[y=0.5cm, x=0.75cm, thick]
  \newcommand{\height}{6}

  \lifelineCompact{0}{$t_{1,3}$}
  \lifelineCompact{2}{$s_1$}
  \lifelineCompact{4}{$s_2$}
  \lifelineCompact{6}{$v$}
  \lifelineCompact{8}{$s_3$}
  \lifelineCompact{10}{$s_4$}
  \lifelineCompact{12}{$t_{2,4}$}

  \def\yAxisCoord{-1.5}
  \draw[->] (\yAxisCoord,\height - .5) -- (\yAxisCoord, -0.5) node[below] {Time};
  \foreach \y in {1,3,5}
  \draw (\yAxisCoord + 0.1,\height - \y) -- (\yAxisCoord -0.1,\height - \y) node[left] {\sbTime{\y}{00}};

    \draw[->, draw=color2] (2, \height - 1) -- (4, \height - 2);
  \draw[->, draw=color2] (4, \height - 3) -- (6, \height - 4);
  \draw[->, draw=color2] (6, \height - 5) -- (12, \height - 6);
  \draw[->, draw=color3] (10, \height - 1) -- (8, \height - 2);
  \draw[->, draw=color3] (8, \height - 3) -- (6, \height - 4);
  \draw[->, draw=color1] (6, \height - 5) -- (0, \height - 6);
 \end{tikzpicture}
}

Consider the network with four commodities in \Cref{fig:heuristic-eps-direction}:
\begin{figure}[htb]
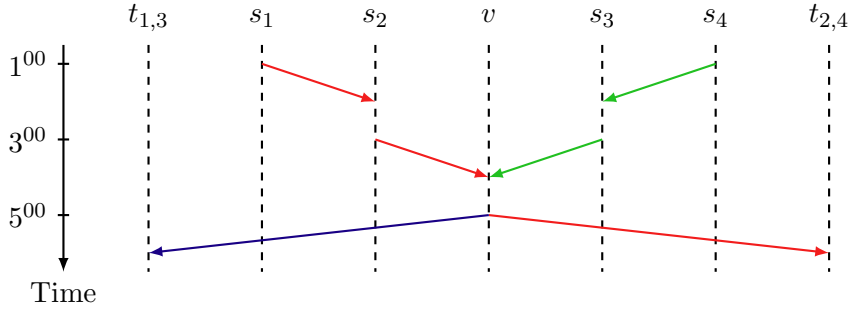

            \centering
        \smallsteps
        \caption{A network illustrating the occurrence of small step sizes in the heuristic}\label{fig:heuristic-eps-direction}
    \end{figure}
Commodity $i\in\{1,2,3,4\}$ has origin $s_i$ and destination $t_i$, and it has exactly one \qq{good} path and an outside-option not displayed in the figure.
There are three vehicles: the \colourone, the \colourthree\oxc and the \colourtwo vehicle.
Each commodity has a demand of $1$. The \colourone and the \colourtwo vehicles have a capacity of $1$ while the \colourthree vehicle has a capacity of $1-\varepsilon$ for some small $\varepsilon > 0$.

We force the selection of the first two feasible directions in the heuristic.
In the initial flow, all particles are sent along their outside option.
In the first iteration, we choose to move particles of commodity~$3$ from their outside option to their shorter path.
We can change our initial flow by $\lambda^{(1)} = 1-\varepsilon$ along this direction.
Secondly, we choose to move particles of commodity~$2$ from their outside option to their shorter path, and change the flow along this direction by $\lambda^{(2)}= 1$.

In iteration $3$, the only commodity that can move to their shorter path is commodity~$1$.
The corresponding direction is, however, not feasible, as commodity~$2$ would use an overfilled boarding edge.
Thus, applying \Cref{alg:remove-overfilled-edges} removes particles from the shorter path of commodity~$2$ at a rate of $1$.
Because the \colourone edge has a remaining capacity of $\varepsilon$, the flow can only be changed along this  resulting direction by $\lambda^{(3)}=\varepsilon$.

In iteration $4$, the only commodity that can move to their shorter path now is commodity~$4$:
By removing $\varepsilon$ flow from commodity~$2$ in the previous iteration, the \colourtwo $v$-$t_{2,4}$-edge has a remaining capacity of $\varepsilon$.
Again, the direction of moving particles of commodity~$4$ to their shorter route is not feasible, and hence by \Cref{alg:remove-overfilled-edges}, particles of commodity~$3$ are removed as well.
The flow is changed along this direction by again $\lambda^{(4)}=\varepsilon$.

Note that after that, we are faced with the same scenario as before iteration $3$; the only difference is that an $\varepsilon$ of flow was shifted away from the good paths of commodities~$2$ and $3$ to the good paths of commodities~$1$ and $4$.
Thus, the same directions of iterations~$3$ and $4$ are applied repeatedly until both commodities 1 and 4 are entirely sent along their good paths.
Hence, the heuristic needs at least $2\cdot (1-\varepsilon)/\varepsilon$ iterations before termination.
\end{example}

\begin{example}[A non-terminating cycle]

The heuristic might also run into cycles from which it does not recover.
A simple network demonstrating this behaviour is shown in \Cref{fig:cyclic-behavior}. All three vehicles have a capacity of $1$, and all three commodities have a demand of $1$. The (non outside-option) paths are displayed to the right.
Note that the flow that sends one unit along $p_2$ and all other particles along their outside option is a user equilibrium.

\newcommand{\inlinetikz}[1]{\begin{tikzpicture}[baseline,every node/.append style={anchor=base, text depth=0}, thick]%
 #1
\end{tikzpicture}}

\newcommand{\cyclingheuristic}{
 \begin{minipage}{0.55\textwidth}
  \centering
  \begin{tikzpicture}[yscale=0.25, xscale=0.4 , thick]
      \newcommand{\height}{12}

      \lifelineCompact{0}{$s_1$}
      \lifelineCompact{2}{$s_2$}
      \lifelineCompact{4}{$s_3$}
      \lifelineCompact{6}{$u$}
      \lifelineCompact{8}{$v$}
      \lifelineCompact{10}{$t_1$}
      \lifelineCompact{12}{$w$}
      \lifelineCompact{14}{$t_{2,3}$}
      
      \def\yAxisCoord{-1.5}
      \draw[->] (\yAxisCoord,\height - .5) -- (\yAxisCoord, -0.5) node[below] {Time};
      \foreach \y in {1,4,7, 10}
      \draw (\yAxisCoord + 0.1,\height - \y) -- (\yAxisCoord -0.1,\height - \y) node[left] {\sbTime{\y}{00}};

            \draw[->, draw=color3] (0, \height - 1) -- (2, \height - 2);
      \draw[->, draw=color3] (2, \height - 2) -- (6, \height - 4);
      \draw[->, draw=color3] (6, \height - 4) -- (8, \height - 5);
      \draw[->, draw=color1] (4, \height - 5) -- (8, \height - 7);
      \draw[->, draw=color1] (8, \height - 7) -- (10, \height - 8);
      \draw[->, draw=color1] (10, \height - 8) -- (12, \height - 9);
      \draw[->, draw=color2] (6, \height - 8) -- (12, \height - 11);
      \draw[->, draw=color2] (12, \height - 11) -- (14, \height - 12);
  \end{tikzpicture}%
\end{minipage}%
\begin{minipage}{0.45\textwidth}%
  \small\vspace{-1em}%
  \begin{align*}
      p_1 & = \inlinetikz{
          \node (s_1) {$s_1$};
          \node (s_2) [right of=s_1] {$s_2$};
          \node (u) [right of=s_2] {$u$};
          \node (v) [right of=u] {$v$};
          \node (t_1) [right of=v] {$t_1$};
          \draw[->, color3] (s_1) -- (s_2);
          \draw[->, color3] (s_2) -- (u);
          \draw[->, color3] (u) -- (v);
          \draw[->, color1] (v) -- (t_1);
      }                    \\
      p_2 & = \inlinetikz{
          \node (s_2) {$s_2$};
          \node (u) [right of=s_2] {$u$};
          \node (v) [right of=u] {$v$};
          \node (t_1) [right of=v] {$t_1$};
          \node (w) [right of=t_1] {$w$};
          \node (t_23) [right of=w] {$t_{2,3}$};
          \draw[->, color3] (s_2) -- (u);
          \draw[->, color3] (u) -- (v);
          \draw[->, color1] (v) -- (t_1);
          \draw[->, color1] (t_1) -- (w);
          \draw[->, color2] (w) -- (t_23);
      }                    \\
      p_3 & = \inlinetikz{
          \node (s_2) {$s_2$};
          \node (u) [right of=s_2] {$u$};
          \node (w) [right of=u] {$w$};
          \node (t_23) [right of=w] {$t_{2,3}$};
          \draw[->, color3] (s_2) -- (u);
          \draw[->, color2] (u) -- (w);
          \draw[->, color2] (w) -- (t_23);
      }                    \\
      p_4 & = \inlinetikz{
          \node (s_3) {$s_3$};
          \node (v) [right of=s_3] {$v$};
          \node (t_1) [right of=v] {$t_1$};
          \node (w) [right of=t_1] {$w$};
          \node (t_23) [right of=w] {$t_{2,3}$};
          \draw[->, color1] (s_3) -- (v);
          \draw[->, color1] (v) -- (t_1);
          \draw[->, color1] (t_1) -- (w);
          \draw[->, color2] (w) -- (t_23);
      }  \end{align*}%
 \end{minipage}%
}

\begin{figure}[ht]
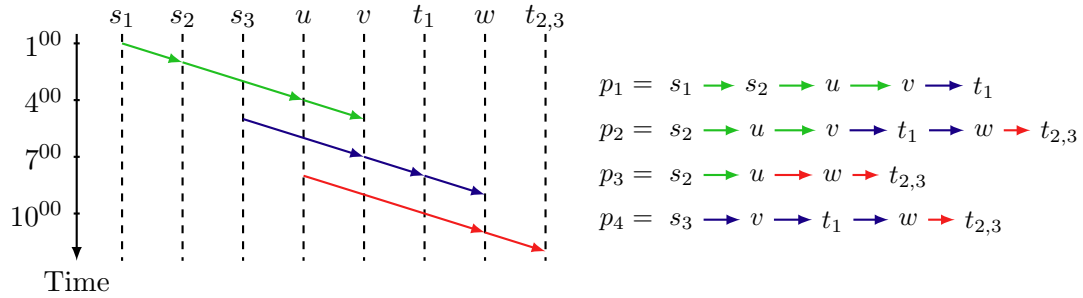

            \centering
        \cyclingheuristic
        \caption{An example for infinite cyclic behaviour of the heuristic}\label{fig:cyclic-behavior}
    \end{figure}

If, however, the heuristic first chooses to fill $p_1$, it ends up in a cycle:
After changing along this direction, only path $p_4$ is not blocked.
Thus, the heuristic chooses to fill path $p_4$ and inevitably removes the flow from $p_1$ again.
Once $p_4$ is filled, the only non-blocked path is $p_3$, and thus $p_3$ is filled and $p_4$ is emptied again.
Finally, $p_1$ is again the only non-blocked path and after filling $p_1$ (and thus removing the flow from $p_3$), we end up in the same situation as after the first augmentation.

\end{example}

\subsubsection*{Reducing the instance complexity}

We employ a technique to reduce the initial complexity of a given instance:
There is a class of paths from which the procedure will never remove flow.
Thus, we first fill these paths directly when initializing the heuristic.
For that we define so-called fixed initial solutions, which represent the prerouted part of the flow.

\newcommand{\ins}{\mathrm{in}}
\newcommand{\outs}{\mathrm{out}}
\begin{definition}
    Let $f=f^{\ins}+f^{\outs}$ be a feasible flow such that $f^{\ins}$ vanishes on the outside options and $f^{\outs}$ vanishes on $\pathsWithoutOutside$.
    Let $f^{(k)}$ denote the flow after iteration $k$ of the procedure given the initial flow $f^{(0)} = f$.
    We call $f^{\ins}$ a \emph{fixed initial solution}, if $f^{\ins}\leq f^{(k)}$ holds component-wise for all $k\in\N$ (for any implementation of the heuristic).
\end{definition}

\begin{definition}
    Let $f^{\ins}$ be a fixed initial solution.
    A path $p\in\pathsWithOutside_i$ is called \emph{uninterruptible} for $f^{\ins}$ if it never boards a vehicle for which the previous driving edge $\tilde e$ of the same vehicle is not saturated, i.e., $f^{\ins}_{\tilde e} < \nu_{\tilde e}$.
\end{definition}
{}
In other words, $p$ is uninterruptible for $f^{\ins}$ if whenever $p$ boards a vehicle at some departure node~$v$ and there is a chain of a driving edge $\tilde e$ and a dwelling edge with head $v$, we must have $f^{\ins}_{\tilde e} = \nu_{\tilde e}$.

\begin{restatable}{proposition}{PropFixedInitialSolutionUninterruptible}\label{prop:fixed-initial-solution-uninterruptible}
    Let $f^{\ins}$ be a fixed initial solution with corresponding feasible flow $f=f^{\ins}+f^\outs$, and let $p\in\pathsWithOutside_i$ be an uninterruptible path of a commodity $i$ with $\sum_{q\in\pathsWithOutside_i} f^\ins_q < \demand_i$.
    If $p$ minimizes $\pi_{i,p}$ on $\{ q \in \pathsWithOutside_i \mid \forall e\in E_B\cap q: f_{e^+} < \nu_{e^+} \}$ and $\pi_{i,p} < \pi_{\outopt_i}$, then $f^\ins + \lambda\cdot 1_{i,p}$ is a fixed initial solution for $\lambda = \min\{ \nu_{e^+} - f_{e^+} \mid e\in E_B\cap p \}\cup\{ \demand_i - \sum_{q\in\pathsWithOutside_i} f^{\ins}_q \}$.
\end{restatable}
\newcommand{\proofPropFixedInitialSolutionUninterruptible}[1][Proof]{
\begin{proof}[#1]
    Let $g$ be the flow that sends the remaining particles of $f^\ins + \lambda \cdot 1_{i,p}$ onto the outside options, i.e., $g_{i,\outopt_i}\coloneqq \demand_i - \sum_{q\in\pathsWithOutside_i}(f^\ins+\lambda \cdot 1_{i,p})_q$.
    Then, $f^*\coloneqq (f^\ins +\lambda \cdot 1_{i,p}) + g$ clearly is a feasible flow.

    The vector $1_{i,p} - 1_{i,\outopt_i}$ is a possible choice for the direction in the first iteration of the procedure given the initial flow $f$ because
    $f_{i,\outopt_i} > 0$ and $p\in \avPaths_{i,\outopt_i}(f)$ hold.
    This direction is also feasible due to \Cref{prop:feasible-direction-boarding-edge}.
    This means, $f+\lambda (1_{i,p} - 1_{i,\outopt_i})$ might be the flow after the first iteration, and thus, as $f^{\ins}$ was a fixed initial solution, $f^{\ins}\leq f^{(k)}$ for all $k\in\N$ where $f^{(k)}$ denotes the flow after $k$ iterations given the initial flow $f+\lambda (1_{i,p} - 1_{i,\outopt_i})$.
    Assume now that there is some $k\in\N$ such that $f^{\ins}_{i,p} + \lambda > f^{(k)}_{i,p}$, and let $k$ be minimal with this property.
    Then, in iteration $k$, some positive amount of flow was removed from entry $(i,p)$.
    This can have two reasons:

    The first reason could be that there is some better path $q\in \pathsWithOutside_i$ with $q\in \avPaths_{i,p}(f^{(k-1)})$, i.e., for any boarding edge $e$ on $q$ for which $e^+$ is not on $p$ we have $f^{(k-1)}_{e^+} < \nu_{e^+}$.
    This, however, implies that for all $e\in E_B\cap q$ we have $f^\ins_{e^+} < \nu_{e^+}$ and thus $\pi_{i,q} < \pi_{i,p}$ contradicts the minimality of $p$.

    The second reason could be that $p$ is removed when establishing feasibility of the direction in~\Cref{alg:remove-overfilled-edges}.
    This can only happen if there is a boarding edge of $p$ for which $f^{(k-1)}_{e^+} = \nu_{e^+}$ and $d_{e^+} > 0$ hold, where $d$ is the direction in iteration~$k-1$.
    This, however, is only possible if the initially chosen direction $1_{i,q'} - 1_{i,p'}$ fulfils $e^+\in q'\setminus p'$.
    As the boarding edge $e$ is already saturated after iteration $(k-1)$, $q'$ must have used the previous driving edge of the same vehicle as $e$.
    As $p$ is uninterruptible for $f^\ins$ (and $f^\ins \leq f^{(k-1)}$), this edge is also already saturated, and thus $p'$ must be a path with $f^\ins_{i,p'} > 0$.
    This implies that $\lambda'\cdot (1_{i,q'} - 1_{i,p'})$ is a possible choice for the first iteration given initial flow $f$ (for some $\lambda > 0$) and therefore contradicts that $f^\ins$ is a fixed initial solution.
\end{proof}
}
\proofPropFixedInitialSolutionUninterruptible
In our implementation, we initialize the flow by filling uninterruptible paths until no more uninterruptible paths exist.
In our computational study in \Cref{sec:comp-study}, this handles between 10\% and 25\% of the total demand before entering the main loop of the heuristic.
Additionally, we initialize the remaining demand not assigned to uninterruptible paths by solving a linear program minimizing the social cost such that no capacity is exceeded.

\section{Computational Study}\label{sec:comp-study}

To gain insights into the applicability of the proposed heuristic, we conduct a computational study on real world train networks.
We analyse the performance of the heuristic and compare the computed equilibrium solutions with system optima.
We consider both the case of departure time choice (DTC) and the case of fixed departure times (FDT).

\subsection{Experiment Setup}

Our dataset provides schedule-based transit networks in the form of periodic schedules.
These periodic schedules are unrolled into a time-expanded transit network as described in \Cref{sec:model} covering the vehicle trips of a typical work day from 5~a.m.\ until 11~p.m.
In the dataset, the demands are given as expected values for each origin-destination pair within one period of the schedule.
For each such pair, we generate a commodity for every 10-minute interval. 
We scale the demand values for these commodities such that they follow the distribution of the travel demand of a typical work day, for which we use the distribution for the Swiss national public transportation network for the year~2000 developed by \citet{Vrtic2007}, as shown in \Cref{fig:dynamic-profile}.

\begin{figure}[b]
    \newcommand{\figureContent}{    
        \begin{tikzpicture}
            \begin{axis}[
                height=4cm,
                width=6.67cm,
                xmin=0, xmax=24,
                ybar,
                bar width=1,
                xtick={0,...,24},
                xticklabels={0,,,,4,,,,8,,,,12,,,,16,,,,20,,,,24},
                xlabel={Hour of day},
                ylabel style={align=center},
                ylabel={Demand share\\(in \%)},
                label style={font=\small},
                tick label style={font=\small},
                enlarge x limits={abs=2},
            ]
                \addplot[ybar, bar width=1.0] table[x expr=\thisrow{time}+0.5, y=demand_share, col sep=comma] {dp.csv};
            \end{axis}
        \end{tikzpicture}
    }
    \newcommand{\figureCaption}{The demand distribution of trips in public transportation networks over a typical work day in Switzerland in the year 2000 \citep{Vrtic2007}.}
            \centering
        \figureContent
        \caption{\figureCaption}\label{fig:dynamic-profile}
    \end{figure}

For the case (DTC), a commodity's time offset represents the target arrival time.
Hence, we instantiate the set of feasible departure times as $\departureTimes_i \coloneqq \R$ and the cost function with a penalty factor of $1$ for both travel time and early arrival, and a penalty factor of $3$ for late arrival, i.e., \[
    \pi_{i,p}\coloneqq \tau_{p} + \max\{0, T_i - \patharrival_{p}\} + 3\cdot \max\{0, \patharrival_{p} - T_i\}.
\]

For the case (FDT), the time offset represents the fixed departure time{} (i.e., $\departureTimes_i$ is a singleton containing this offset).
Here, we assume that users minimize their travel time, i.e., $\pi_{i,p}=\tau_{i,p}$.
In both cases, the cost of the outside option of all commodities are fixed to a common, instance-specific value.

\newcommand*{\BR}{\pi^*}
\newcommand*{\absReg}{r}
\newcommand*{\relReg}{r^{\text{rel}}}
\newcommand*{\approxFactor}{\rho}

To measure the quality of the computed flows, we consider several metrics:
For a given flow~$f$, if a particle of commodity $i$ uses path $p$, then the particle's \emph{(absolute) regret} is defined as the difference of the cost of $p$ and the minimum cost over all available alternatives in $A_{i,p}(f)$, i.e., $\absReg_{i,p}(f)\coloneqq \pi_{i,p} - \BR_{i,p}(f)$ with $\BR_{i,p}(f)\coloneqq \min_{q\in A_{i,p}(f)}\pi_{i,q}$.
We define the particle's \emph{approximation factor} as the ratio $\approxFactor_{i,p}(f)\coloneqq \pi_{i,p} / \BR_{i,p}(f)$.
Note that a flow is a user equilibrium if and only if all particles have an approximation factor of $1$.
Thus, we use the distribution of the approximation factor to measure how close a flow is to a user equilibrium.
This includes the mean and 99\ordinalth{} percentile (P99) of the distribution.
Here, the P99 approximation factor is defined as the minimum value $v\in\R_{\geq0}$ such that at least 99\% of all particles have an approximation factor of at most~$v$.

For both the (DTC) and the (FDT) scenarios, we first compute a system-optimal flow by solving the
following linear program using a column generation method: 
\begin{align*}
    \min_{f\in\feasFlows}\  & \sum_{(i,p)\in\pathsWithOutside}\pi_{i,p} \cdot f_{i,p}.
\end{align*}
Then, we produce flows using our proposed heuristic.
The computation time of the main loop of the heuristic is limited to at most 2 hours.
The flow with the minimal mean approximation factor discovered during the main loop of the heuristic is used as the final solution.

As the problem changes significantly with a higher demand-capacity ratio, we conduct the same experiments once more by artificially scaling the demand of the network by a factor of $10$.

We implemented the heuristic for computing user equilibria (available in \citealt{HeuristicImplementation}) and the computation of system optima using the Rust programming language.
All experiments were conducted on an AMD Ryzen~9 5950X CPU with Gurobi~11 as LP solver.

\subsection{Data}

A pool of periodic timetables of real-world public transportation networks is provided in the publicly available TimPassLib \citep{SchieweGL23}.
We consider the following networks: a network for regional trains in Lower Saxony, Germany, two networks for the district of Erding including a slice of the Munich S-Bahn, the Hamburg S-Bahn, the Athens Metro and the long-distance train networks for Germany and Switzerland.
As an example, the Hamburg S-Bahn network is displayed in \Cref{fig:map-sbahn-hamburg}.{}

In some instances of this dataset, the demand values are unrealistically high, in which case we scale the demand values down to a more realistic level.
Similarly, in some instances, concrete numbers of vehicle capacities are not provided, in which case we use realistic estimates.
\newcommand{\datasetSpecifics}{
More specifically, for the Lower Saxony Regional network, we scaled the demand down to a nominal demand of $\num{5215488}$ trips per day, close to numbers reported by \citet{NiedersachsenDemand} in 2023, and vehicles with a capacity of $\num{500}$ passengers.
For the Erding networks, the original dataset \citep{PublicTransportNetworks2018} prescribes a demand of $\num{6700}$ per hour, which we multiplied by 16 hours for a total demand of $\num{107200}$.
These networks consist of lines operated by S-Bahn trains with a capacity of $\num{1000}$ passengers and bus lines with a capacity of $\num{70}$ passengers.
For the Hamburg S-Bahn network, we use as nominal demand (in trips per day) a value of $\num{750000}$ as taken from~\citep{sbahn-hamburg-data}, and vehicles with a capacity of $\num{1000}$ passengers.
For the Athens Metro network, the original dataset represents the peak demand, which we scale by a factor of $8$ hours to obtain a nominal demand of $\num{3039504}$ trips per day.
}
\datasetSpecifics\par\Cref{table:considered-networks} describes the considered networks and schedules in more detail.

\begin{figure}[t]
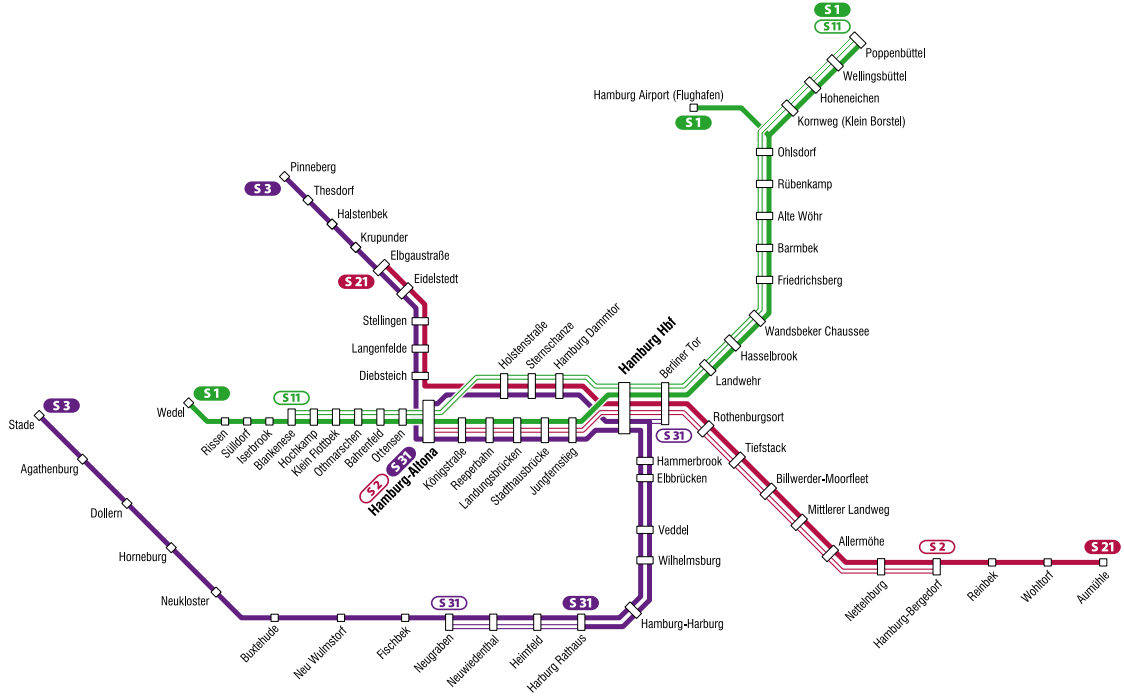

    \centering
    \iftrue
        \includesvg[width=\textwidth]{Karte_der_S-Bahn_Hamburg_no_river.svg}
        \vspace{-2em}
        \caption{Map of the considered S-Bahn Hamburg as operated until 2023 \citep{MapHamburg}}
        \label{fig:map-sbahn-hamburg}
    \else
        \FIGURE{\includesvg[width=\textwidth]{Karte_der_S-Bahn_Hamburg_no_river.svg}}{Map of the considered S-Bahn Hamburg as operated until 2023
        \label{fig:map-sbahn-hamburg}\citep{MapHamburg}}{}
    \fi
\end{figure}

\newcommand{\tabularinstances}{
    \small
    \begin{tabular}{lrrrrr}
        \toprule
        Name                  & \# stations    & \# vehicles   & nominal demand  & \# commodities & $\pi_{\outopt}$  \\\midrule
        Lower Saxony Regional & $\num{34}$     & $\num{468}$   & $\num{5215488}$ & $\num{31680}$  & 180  \\
        Erding\_NDP\_S020     & $\num{51}$     & $\num{1728}$  & $\num{107200}$ & $\num{64800}$   & 180 \\
        Erding\_NDP\_S021     & $\num{51}$     & $\num{1728}$  & $\num{107200}$ & $\num{64800}$   & 180 \\
        Hamburg S-Bahn        & $\num{68}$     & $\num{1512}$  & $\num{750000}$  & $\num{194880}$ & 180 \\
        Athens Metro          & $\num{51}$     & $\num{2160}$  & $\num{3039504}$ & $\num{228960}$ & 180 \\
        German Long Distance  & $\num{250}$    & $\num{1512}$  & $\num{6173888}$ & $\num{586176}$ & 360 \\
        Swiss Long Distance   & $\num{140}$    & $\num{1540}$  & $\num{1347686}$ & $\num{1159872}$ & 180 \\
        \bottomrule
    \end{tabular}
}

\begin{table}[ht]
            \centering
        \caption{Details of the considered networks}\label{table:considered-networks}
        \tabularinstances
        \vspace{0.5em}
    \end{table}

\FloatBarrier
\subsection{Results}

The results of our computational study are summarized in \Cref{table:results-both-settings-nominal} for the nominal demand case and in \Cref{table:results-both-settings-scaled} for the scaled demand case.
The tables show the following three metrics for both the flow produced by the heuristic and the system optimum:
The mean approximation factor, the P99 approximation factor\oxc and the percentage $S_{r=0}$ of particles that have no regret.

\newcommand*{\Cell}[2]{\tiny$\num{#1}$&\tiny$\num{#2}$}%
\newcommand{\hint}[1]{\textcolor{gray}{\footnotesize #1}}%

\newlength{\sepcols}
\setlength{\sepcols}{0.01cm}\newcommand\networkspec{\tiny}

\begin{table}    \newcommand{\tableContent}{
    \footnotesize    \begin{tabular}{l p{\sepcols} r@{~~}rr@{~~}rr@{~~}r p{\sepcols} r@{~~}rr@{~~}rr@{~~}r}
            \toprule
            && \multicolumn{6}{c}{(FDT)} & & \multicolumn{6}{c}{(DTC)} 
            \\\addlinespace
            Network               && \multicolumn{2}{c}{mean $\approxFactor$} & \multicolumn{2}{c}{P99 $\approxFactor$} & \multicolumn{2}{c}{$S_{r=0}$} && \multicolumn{2}{c}{mean $\approxFactor$} & \multicolumn{2}{c}{P99 $\approxFactor$} & \multicolumn{2}{c}{$S_{r=0}$} \\
            \midrule
            \networkspec Lower Saxony Regional & & \Cell{1.007}{1.190} & \Cell{1.235}{6.120} & \Cell{98.4}{92.9} & & \Cell{1.136}{1.149} & \Cell{4.250}{4.727} & \Cell{89.2}{91.0} \\
            \networkspec  Erding\_NDP\_S020     & & \Cell{1.000}{1.011} & \Cell{1.000}{1.225} & \Cell{100.0}{98.8} & & \Cell{1.005}{1.009} & \Cell{1.075}{1.078} & \Cell{98.8}{98.8}\\
            \networkspec Erding\_NDP\_S021     & & \Cell{1.000}{1.008} & \Cell{1.000}{1.031} & \Cell{100.0}{99.0} & &  \Cell{1.005}{1.006} & \Cell{1.000}{1.182} & \Cell{99.0}{98.5} \\ 
            \networkspec Hamburg S-Bahn        & & \Cell{1.000}{1.000} & \Cell{1.000}{1.000} & \Cell{100.0}{100.0} & & \Cell{1.000}{1.000} & \Cell{1.000}{1.000} & \Cell{100.0}{99.9} \\ 
            \networkspec Athens Metro          & & \Cell{1.736}{1.749} & \Cell{7.030}{7.317} & \Cell{70.6}{72.1} & & \Cell{1.734}{1.762} & \Cell{6.983}{7.143} & \Cell{67.1}{70.3} \\ 
            \networkspec German Long Distance  & & \Cell{1.120}{1.184} & \Cell{2.117}{2.463} & \Cell{63.1}{57.1} & & \Cell{1.106}{1.161} & \Cell{2.151}{2.479} & \Cell{72.2}{68.4} \\
            \networkspec Swiss Long Distance   & & \Cell{1.000}{1.049} & \Cell{1.000}{2.378} & \Cell{99.8}{90.3} & & \Cell{1.032}{1.040} & \Cell{1.776}{2.093} & \Cell{93.4}{91.8} \\
            \bottomrule
    \end{tabular}
    }
    \newcommand{\tableCaption}{Resulting performance metrics with nominal demand. The first number in every cell belongs to the flow computed by the heuristic, the second number to the system optimum.}

            \centering
        \caption{\tableCaption}\label{table:results-both-settings-nominal}
        \tableContent
    \end{table}
\begin{table}
    \newcommand{\tableContent}{
    \footnotesize    \begin{tabular}{l p{\sepcols} r@{~~}rr@{~~}rr@{~~}r p{\sepcols} r@{~~}rr@{~~}rr@{~~}r}
            \toprule
            && \multicolumn{6}{c}{(FDT)} & & \multicolumn{6}{c}{(DTC)} 
            \\\addlinespace
            Network               && \multicolumn{2}{c}{mean $\approxFactor$} & \multicolumn{2}{c}{P99 $\approxFactor$} & \multicolumn{2}{c}{$S_{r=0}$} && \multicolumn{2}{c}{mean $\approxFactor$} & \multicolumn{2}{c}{P99 $\approxFactor$} & \multicolumn{2}{c}{$S_{r=0}$} \\
            \midrule
            \networkspec Lower Saxony Regional & & \Cell{1.000}{2.621} & \Cell{1.000}{12.000} & \Cell{100.0}{65.9} & & \Cell{1.061}{2.590} & \Cell{2.792}{12.857} & \Cell{95.2}{67.3} \\ 
            \networkspec Erding\_NDP\_S020     & & \Cell{1.000}{2.334} & \Cell{1.000}{6.882} & \Cell{100.0}{53.5} & & \Cell{1.013}{3.165} & \Cell{1.343}{7.200} & \Cell{98.1}{54.3} \\
            \networkspec Erding\_NDP\_S021      & & \Cell{1.000}{3.169} & \Cell{1.000}{7.500} & \Cell{100.0}{52.5} & & \Cell{1.008}{2.836} & \Cell{1.184}{7.500} & \Cell{98.5}{54.9} \\
            \networkspec Hamburg S-Bahn       & & \Cell{1.656}{2.921} & \Cell{8.167}{18.500} & \Cell{57.3}{57.5} & & \Cell{1.900}{3.137} & \Cell{8.947}{21.000} & \Cell{59.0}{58.0} \\
            \networkspec Athens Metro         & & \Cell{1.166}{3.434} & \Cell{4.896}{14.062} & \Cell{88.7}{38.0} & & \Cell{1.222}{3.792} & \Cell{5.767}{18.367} & \Cell{88.8}{36.6} \\
            \networkspec German Long Distance    & & \Cell{1.227}{1.389} & \Cell{2.823}{3.636} & \Cell{52.4}{40.3} & & \Cell{1.112}{1.338} & \Cell{2.494}{3.913} & \Cell{76.7}{56.3} \\
            \networkspec Swiss Long Distance    & & \Cell{1.000}{1.224} & \Cell{1.000}{4.286} & \Cell{100.0}{76.4} & & \Cell{1.019}{1.233} & \Cell{1.473}{4.000} & \Cell{97.3}{77.1} \\
            \bottomrule
    \end{tabular}
    }
    \newcommand{\tableCaption}{Resulting performance metrics with 10x demand. The first number in every cell belongs to the flow computed by the heuristic, the second number to the system optimum.}

            \centering
        \caption{\tableCaption}\label{table:results-both-settings-scaled}
        \tableContent
    \end{table}

It can be observed that, for the case of (FDT) with nominal demand, the heuristic produces flows performing significantly better with respect to the regret metrics than the system optimum flow, except in the case of the Athens Metro network.
In particular, for the Erding and Hamburg network, user equilibria with no regret at all are found; besides these also for the Swiss Long Distance network a flow with a P99 approximation factor of $1$ is found.
The social cost of the flow of the heuristic is at most $7.6\%$ higher than that of the system optimum flow except for the Lower Saxony regional network, where it is $20.4\%$ higher.

For the departure time choice scenario with nominal demand, the heuristic in most cases produces flows that are slightly better than the system optimum with respect to the regret metrics.
Only for one network, the heuristic could produce an exact (up to machine precision) user equilibrium.
It is worth noting that the P99 approximation factor is smaller than $2.5$ for $5$ out of $7$ networks for both the flow produced by the heuristic and the system optimum.
Further, the social cost is at most $6.4\%$ higher than that of the system optimum flow.

For the scaled demand cases, the results of the heuristic and the system optimum differ more significantly.
In particular, for the (FDT) scenario, the heuristics computes exact user equilibria for four out of seven networks; for two of the remaining three networks, it produces a P99 approximation factor of less than halve compared to the system optimum flow.
For the (DTC) scenario, the P99 approximation factor is smaller by a factor of $5$ compared to the system optimum flow for two networks, and smaller by a factor of more than $2$ for six out of seven networks.

\FloatBarrier

\section{Conclusion}
We presented a side-constrained user equilibrium model for a schedule-based transit network incorporating hard vehicle capacities.
As our main results, we proved that equilibria exist for fixed departure times and that they can be computed efficiently for single-commodity instances.
The existence result generalizes a classical result of \citet{BernsteinS94}; its proof is based on a new condition (weak regularity) implying existence of BS-equilibria for a class of discontinuous and non-separable cost maps.
For general multi-commodity instances we showed hardness results and devised a heuristic, which was implemented and tested on several realistic transportation networks.

\paragraph*{Open Problems.}
Firstly, a side-constrained user equilibrium is not unique and, hence, the issue of equilibrium selection or determining which equilibrium is likely to be observed in practice deserves further study.
From an algorithmic point of view, while we proved the NP-hardness in the multi-commodity setting, it is unclear whether these hardness results also apply to the single-commodity case{} and whether the described decision problems lie in NP.
Another open problem is to determine the complexity of computing a user equilibrium for multi-commodity networks with fixed departure times in the sense of a total function problem rather than a decision problem.
Similarly, the computational complexity for single-commodity networks
with periodic timetables and a compactly representable time-expanded graph is also open.

\newcommand{\listOfSymbols}{
    The following list contains the symbols used to model side-constrained user equilibria in schedule-based transit networks in this paper.
    Symbols that are used only in single subsections are at the end of the list.

    \medskip\noindent
{
    \small
    \begin{tabular}{p{.18\textwidth}p{.82\textwidth}}%
        \toprule
        Symbol & Description \\ \midrule

        $\stations$ & set of stations \\
        $\vehicles$ & set of vehicle trips \\
        $G=(V, E)$ & time-expanded graph with nodes $V$ and edges $E$ \\
        $E_B \subseteq E$ & set of boarding edges \\
        $E_D \subseteq E$ & set of driving edges \\
        $e^+$ & succeeding driving edge of a boarding edge $e$ \\
        $\mu_z, \mu_e$ & capacity of a vehicle $z$ and a driving edge $e$ \\
        $\theta(v)$ & time of a node $v$ \\
        $\tau_e, \tau_p$ & traversal time of edge $e$ and path $p$ \\
        $I$ & set of commodities \\
        $s_i, t_i$ & origin and destination station of commodity $i$ \\
        $\departureTimes_i$ & set of feasible departure times of commodity $i$ \\
        $T_i$ & target arrival time of commodity $i$ \\
        $\demand_i$ & total demand of commodity $i$ \\
        $\pathsWithoutOutside_i$ & set of paths of commodity $i$ \\
        $\outopt_i$ & outside option of commodity $i$ \\
        $\pathsWithOutside_i$ & strategy set of commodity $i$: $\pathsWithOutside_i \coloneqq \pathsWithoutOutside_i \cup \{\outopt_i\}$ \\
        $\beta_i, \gamma_i^+, \gamma_i^-$ & penalty factors of commodity $i$ for travel time, late arrival\oxc and early arrival \\
        $\pi_{i,p}$ & cost of strategy $p\in\pathsWithOutside_i$ of commodity $i$ \\

        $f_{i,p}$ & flow of commodity $i$ on path $p$ \\
        $\demFeasFlows,\capFeasFlows,\feasFlows$ & sets of demand-feasible, capacity-feasible and feasible flows \\
        $f_{i,p\to q}(\varepsilon)$ & $\varepsilon$-deviation of commodity $i$ from path $p$ to path $q$ \\
        $\avPaths_{i,p}(f)$ & set of available alternatives to path $p$ for commodity $i$ given flow $f$ \\
        $\pi(f)$ & social cost of flow $f$ \\[1em]

        \textsc{\Cref{sec:model}} \\
        $J$ & set of groups of particles with common penalty factors, origin and destination, target arrival time\oxc and feasible departure time interval \\
        $\elasticDemand_j(\pi)$ & volume of particles of a group $i\in J$ willing to travel if the cost does not exceed~$\pi$\\[1em]

        \textsc{\Cref{sec:characterization-existence-pos}} \\
        $\admissibleDeviations(f)$ & set of admissible $\varepsilon$-deviations of flow $f$ \\
        $c_{i,p}$ & cost function of path $p$ for commodity $i$ \\
        $G'=(V', E')$, $\alpha_i$, $\omega_i$ & expansion graph of $G$ with source nodes $\alpha_i$ and sink nodes $\omega_i$ \\
        $\bar c_e, \bar c_{i,e}$ & upper hull of a cost function $c_e$ or $c_{i,e}$ (see \Cref{def:regular-cost-function}) \\[1em]

        \textsc{\Cref{sec:computation}} \\
        $p \prec q$ & path $p$ has priority over path $q$ \\
        $\pi'_i(\theta)$ & cost when arriving at time $\theta$ (given fixed departure times) \\

        \bottomrule
    \end{tabular}
    }
}

\clearpage

    \appendix
    \section{List of Symbols}\label{sec:list-of-symbols}
    \listOfSymbols
    \clearpage
    \printbibliography

\end{document}